\documentclass{IEEEojcsys}
    
    \usepackage{color,array}
    \usepackage{jhoagg}
    \usepackage{mathtools}
    \usepackage{amsfonts}
    \usepackage{amssymb,latexsym}
    \usepackage[psamsfonts]{eucal}
    \usepackage{amsthm}
    \usepackage{graphicx}
    \let\labelindent\relax 
    \usepackage[inline]{enumitem}
    \usepackage{indentfirst}
    \usepackage{setspace}
    \usepackage{microtype}
    \usepackage{threeparttable}
    \usepackage{float}
    \usepackage{afterpage}
    \usepackage{placeins}
    \usepackage{cancel}
    \usepackage{leftidx}
    \usepackage{breqn}
    \usepackage[export]{adjustbox}
    \usepackage{comment}
    \usepackage[ruled, linesnumbered, lined]{algorithm2e}
    \usepackage[dvipsnames]{xcolor}
    \usepackage{xcolor}
    \usepackage[most]{tcolorbox}
    \usepackage{soul}
    \usepackage[noadjust]{cite}
    \usepackage{standalone}

\usepackage{hyperref}

\makeatletter
\hypersetup{colorlinks=true}
\AtBeginDocument{\@ifpackageloaded{hyperref}
  {\def\@linkcolor{blue!60!black}
   \def\@anchorcolor{blue!60!black}
   \def\@citecolor{blue!60!black}
   \def\@filecolor{blue!60!black}
   \def\@urlcolor{blue!60!black}
   \def\@menucolor{blue!60!black}
   \def\@pagecolor{blue!60!black}
\begingroup
  \@makeother\`%
  \@makeother\=%
  \edef\x{%
    \edef\noexpand\x{%
      \endgroup
      \noexpand\toks@{%
        \catcode 96=\noexpand\the\catcode`\noexpand\`\relax
        \catcode 61=\noexpand\the\catcode`\noexpand\=\relax
      }%
    }%
    \noexpand\x
  }%
\x
\@makeother\`
\@makeother\=
}{}}
\makeatother

    \usepackage{cleveref}

    \usepackage{tikz}
    \usetikzlibrary{shapes,arrows,positioning,fit}

    \usepackage{pgf}
    \usepackage{pgfplots}
    \usepgflibrary{plothandlers}
    \usepgfplotslibrary{external} 
    \pgfkeys{/pgf/number format/.cd,1000 sep={}}  
    \usetikzlibrary{calc,shapes.misc,plotmarks}
    \pgfplotsset{compat=1.13}
    
    \let\originalleft\left
    \let\originalright\right
    \renewcommand{\left}{\mathopen{}\mathclose\bgroup\originalleft}
    \renewcommand{\right}{\aftergroup\egroup\originalright}
    
    \newcounter{thm} 
    \newtheorem{theorem}[thm]{\indent Theorem}
    
    \newtheorem{assumption}{\indent Assumption}
    
    \newtheorem{proposition}{\indent Proposition}
    
    \newtheorem{lemma}{\indent Lemma}
    
    \newtheorem{remark}{\indent Remark}
    
    \newtheorem{corollary}{\indent Corollary}
    
    \newtheorem{definition}{\indent Definition}
    
    \newtheorem{example}{\indent Example}
    
    \newtheorem{fact}{\indent Fact}
    
    \newtheorem{conjecture}{\indent Conjecture}
    
    \newtheorem{experiment}{\indent Experiment}

    \allowdisplaybreaks

    \newcommand{\bd}[0]{\mbox{bd }}

    \newlist{enumA}{enumerate}{1}
    \setlist[enumA,1]{label=(A\arabic*),leftmargin=1cm}

    \newlist{enumC}{enumerate}{1}
    \setlist[enumC,1]{label=(C\arabic*),leftmargin=1cm}

    		\newcommand\xqed[1]{%
      \leavevmode\unskip\penalty9999 \hbox{}\nobreak\hfill
      \quad\hbox{#1}}
    \newcommand\exampletriangle{\xqed{$\triangle$}}

    \usepackage{cite}
    \usepackage{accents}
    
    \newlength\figureheight 
    \newlength\figurewidth

    \allowdisplaybreaks
    
    \graphicspath{ {Figures/} }

    \DeclareMathAlphabet{\mathcal}{OMS}{cmsy}{m}{n} 

    \crefname{equation}{}{}
    
    \usepackage[ruled, linesnumbered, lined]{algorithm2e}
    \SetKwInput{KwInit}{Initialize}
    \SetKwFunction{SafeSet}{SafeSet}
    \SetKwFunction{GetSafe}{GetSafe}
    \SetKwFunction{GetUnsafe}{GetUnsafe}
    
    \SetCommentSty{mycommfont}
    
    \newlist{enumalph}{enumerate}{1}
    \setlist[enumalph]{label=\textit{(\alph*)}}

    \jvol{00}
    \jnum{XX}
    \paper{XXXX}
    \receiveddate{XX March 2024}

\begin{document}


\title{A Closed-Form Control for Safety Under Input Constraints Using a Composition of Control Barrier Functions} 


\author{PEDRAM RABIEE\affilmark{1}}

\author{JESSE B. HOAGG\affilmark{1}  (Senior Member, IEEE)}

\affil{Department of Mechanical and Aerospace Engineering, University of Kentucky, Lexington, KY 40506 USA} 

\corresp{CORRESPONDING AUTHOR: J. B. Hoagg (e-mail: 
{jesse.hoagg@uky.edu})}
\authornote{This work is supported in part by the National Science Foundation (1849213,1932105) and Air Force Office of Scientific Research (FA9550-20-1-0028).}

\markboth{
A CLOSED-FORM CONTROL FOR SAFETY UNDER INPUT CONSTRAINTS USING A COMPOSITION OF CONTROL BARRIER FUNCTIONS}{P. RABIEE {\itshape ET AL}.}

\begin{abstract}
We present a closed-form optimal control that satisfies both safety constraints (i.e., state constraints) and input constraints (e.g., actuator limits) using a composition of multiple control barrier functions (CBFs).
This main contribution is obtained through the combination of several ideas.
First, we present a method for constructing a single relaxed control barrier function (R-CBF) from multiple CBFs, which can have different relative degrees.
The construction relies on a log-sum-exponential soft-minimum function and yields an R-CBF whose zero-superlevel set is a subset of the intersection of the zero-superlevel sets of all CBFs used in the composition.
Next, we use the soft-minimum R-CBF to construct a closed-form control that is optimal with respect to a quadratic cost subject to the safety constraints.
Finally, we use the soft-minimum R-CBF to develop a closed-form optimal control that not only guarantees safety but also respects input constraints.
The key elements in developing this novel control include: the introduction of the control dynamics, which allow the input constraints to be transformed into controller-state constraints; the use of the soft-minimum R-CBF to compose multiple safety and input CBFs, which have different relative degrees; and the development of a desired surrogate control (i.e., a desired input to the control dynamics).
We demonstrate these new control approaches in simulation on a nonholonomic ground robot.
\end{abstract}


\maketitle

\section{Introduction}\label{sec:introduction}

Control barrier functions (CBFs) are used to determine controls that make a designated safe set forward invariant \cite{ames2016control,xu2015robustness}. 
Thus, CBFs can be used to generate controls that guarantee safety constraints (i.e., state constraints).
CBFs are often integrated into real-time optimization-based control methods (e.g., quadratic programs) as safety filters 
\cite{wabersich2022predictive,xu2017correctness,seiler2021control,breeden2023robust}.
They are also used in conjunction with stability constraints and/or performance objectives \cite{nguyen2016exponential, romdlony2016stabilization}. 
Related barrier functions are used for Lyapunov-like control design and analysis (e.g., \cite{prajna2007framework, panagou2015distributed, tee2009barrier, jin2018adaptive}).
CBF methods have been demonstrated in a variety of applications, including mobile robots \cite{nguyen2015safety, srinivasan2020control, jian2023dynamic, safari2023time}, unmanned aerial vehicles \cite{borrmann2015control, singletary2022onboard}, and autonomous vehicles \cite{ames2016control, seo2022safety, alan2023control}.

One important challenge in CBF methods is to verify a candidate CBF, that is, confirm that the candidate CBF satisfies the conditions to be a CBF \cite{wisniewski2015converse}.
For systems without input constraints (e.g., actuator limits), a candidate CBF can often be verified provided that it satisfies certain structural assumptions (e.g., constant relative degree) \cite{ames2016control}.
In contrast, verifying a candidate CBF under input constraints can be challenging, and this challenge is exacerbated if the safe set is described using multiple candidate CBFs.
It may be possible to use offline sum-of-squares optimization methods to verify a candidate CBF \cite{wang2018,clark2021,isaly2022feasibility,pond2022fast}.  
Alternatively, it may be possible to synthesize a CBF offline by griding the state space \cite{tan2022compatibility}.

An online approach to obtain forward invariance (e.g., state constraint satisfaction) subject to input constraints is to use a prediction of the system trajectory into the future under a backup control. 
For example,~\cite{gurrietScalableSafety2020,chen2020guaranteed} determine a control forward invariant subset of the safe set by using a finite-horizon prediction of the system under a backup control. 
However,~\cite{gurrietScalableSafety2020,chen2020guaranteed} require replacing an original barrier function that describes the safe set with multiple barrier functions---one for different time instants of the prediction horizon.
Thus, the number of barrier functions increases as the prediction horizon increases, which can lead to conservative constraints and result in a set of constraints that are not simultaneously feasible.
These drawbacks are addressed in \cite{rabiee2023softmin,rabiee2023softmax} by using a log-sum-exponential soft-minimum function to construct a single composite barrier function from the multiple barrier functions that arise from using a prediction horizon. 
In addition, \cite{rabiee2023softmax} uses a log-sum-exponential soft-maximum function to allow for multiple backup controls.
The use of multiple backups can enlarge the verified forward-invariant subset of the safe set. 
However, \cite{gurrietScalableSafety2020, chen2020guaranteed, rabiee2023softmin, rabiee2023softmax} all rely on a prediction of the system trajectories into the future.
Another approach to address safety subject to input constraints is presented in \cite{black2023consolidated}, which uses a composition of multiple CBFs, where the composition has adaptable weights. 
However, the feasibility of the update law for the weights is related to the feasibility of the original optimization problem subject to input constraints.
Smooth barrier function compositions (e.g., log-sum-exponential approximation of minimum and maximum) also appear in~\cite{safari2023time,lindemann2018control}, and nonsmooth compositions are in~\cite{glotfelter2017nonsmooth,glotfelter2019hybrid}.

This article presents a new approach to address forward invariance subject to input constraints.
Specifically, we use the soft-minimum function to combine multiple safety constraints (i.e., state constraints) and multiple input constraints (e.g., actuator limits) into a single constraint. 
This composition can yield a relaxed control barrier function (R-CBF) that satisfies the CBF criteria on the boundary of its zero-superlevel set but not necessarily on the interior.
This composite soft-minimum R-CBF is then used in a constrained quadratic optimization to generate a control that is optimal and satisfies both safety and input constraints. 
We present a closed-form control that satisfies the constrained quadratic optimization, thus eliminating the need to solve a quadratic program in real time.
Closed-form solutions to other CBF-based optimizations appear in~\cite{wieland2007constructive,ames2016control,cortez2022compatibility}.
We also note that model predictive control methods can be used to obtain closed-form controls for systems with linear dynamics and polytopic constraints on the state and input~\cite{borrelli2017predictive,bemporad2002model,tondel2003algorithm}. 
However, extending these methods to systems with nonlinear dynamics or general nonlinear constraints typically requires solving a nonlinear optimization problem in real time.
To our knowledge, this article is the first to present a closed-form CBF-based control that satisfies multiple safety constraints as well as multiple input constraints.

This new closed-form optimal control that satisfies both safety constraints and input constraints is obtained through the combination of several ideas.
First, \Cref{section:composite CBF} presents a method for using the soft-minimum function to construct a single R-CBF from multiple CBFs, where each CBF in the composition can have different relative degree. 
The zero-superlevel set of this R-CBF is a subset of the intersection of the zero-superlevel sets of all the CBFs used in the construction.
Next, \Cref{section:control} uses the soft-minimum R-CBF and the introduction of a slack variable to construct a closed-form optimal control that guarantees safety. 
The control is optimal with respect to a quadratic performance function subject to safety constraints (i.e., state constraints).
The method is demonstrated on a simulation of a nonholonomic ground robot subject to position and speed constraints, which do not have the same relative degree.

\Cref{section:input constraints} presents the main contribution of this article, namely, a closed-form optimal control that not only guarantees safety (i.e., state constraints) but also respects input constraints (e.g., actuator limits).
To do this, we introduce control dynamics where the control signal is an algebraic function of the controller state, and the input to the control dynamics (i.e., the surrogate control) is the closed-form solution to a constrained optimization.
The use of control dynamics allows us to express the input constraints as CBFs in the state of the controller.
Notably, these input-constraint CBFs do not have the same relative degree as the safety-constraint CBFs. 
However, this difficult is addressed using the composite soft-minimum R-CBF construction. 
Lastly, we introduce a desired surrogate control (i.e., the desired input to the control dynamics) in order to convert the original optimal control problem into a surrogate optimization, where the optimization variable is the surrogate control.
Other methods using control dynamics and CBFs include \cite{xiao2022control,ames2020integral}; however, neither of these address the general problem of generating an optimal control with multiple input constraints and multiple state constraints, which have potentially different relative degree.
For example, \cite{xiao2022control} is limited to one state constraint and constant upper-and-lower-bound input constraints; and it does not address the optimality of the original control input.
Most importantly, \cite{xiao2022control,ames2020integral} use multiple constraints that are not guaranteed to be simultaneously feasible. 
In contrast, this article presents a closed-form optimal control that addresses multiple state and input constraints with feasibility analysis.
We demonstrate this new control method in simulations of a nonholonomic ground robot subject to position constraints, speed constraints, and input constraints---none of which have the same relative degree. 
Some preliminary results related to this article appear in the preliminary conference article \cite{rabiee2023composition}; however, \cite{rabiee2023composition} does not provide a comprehensive analysis (e.g., no proofs are provided).
In addition, this article goes significantly beyond \cite{rabiee2023composition} by presenting the closed-form optimal-and-safe controls; analyzing these closed-form methods; relaxing assumptions; and providing more, discussion, examples, and simulations.

\section{Notation}

The interior, boundary, and closure of $\SA \subseteq \BBR^n$ are denoted by $\mbox{int }\SA$, $\bd\SA$, $\mbox{cl }\SA$, respectively.
Let $\mbox{conv }\SA$ denote the convex hull of $\SA \subset \BBR^n$.
Let $\BBP^{n}$ denote the set of symmetric positive-definite matrices in $\BBR^{n \times n}$.

Let $\eta:\BBR^n \to \BBR^\ell$ be continuously differentiable. 
Then, $\eta^\prime :\BBR^n \to \BBR^{\ell \times n}$ is defined as $\eta^\prime(x) \triangleq \pderiv{\eta(x)}{x}$. 
The Lie derivatives of $\eta$ along the vector fields of $\psi:\BBR^n \to \BBR^{n \times m}$ is defined as
\begin{equation*}
L_\psi \eta(x) \triangleq \eta^\prime(x) \psi(x).
\end{equation*}
If $m=1$, then for all positive integers $d$, define
\begin{equation*}
    L_\psi^d \eta(x) \triangleq L_\psi L_\psi^{d-1} \eta(x).
\end{equation*} 
Throughout this paper, we assume that all functions are sufficiently smooth such that all derivatives that we write exist and are continuous.

A continuous function $a \colon \BBR \to \BBR$ is an \textit{extended class-$\SK$ function} if it is strictly increasing and $a(0)=0$.

Let $\rho>0$, and consider $\mbox{softmin}_\rho : \BBR^N\to \BBR$ defined by 
\begin{equation}\label{eq:softmin}
\mbox{softmin}_\rho (z_1,\ldots,z_N) \triangleq -\frac{1}{\rho}\log\sum_{i=1}^Ne^{-\rho z_i},
\end{equation}
which is the log-sum-exponential \textit{soft minimum}. 
The next result relates the soft minimum to the minimum. 
The proof is in Appendix~\ref{app:input_constraint_propositions}.

\begin{proposition} \label{prop:softmin_bound}
\rm{
Let $z_1,\ldots, z_N \in \BBR$.
Then,
\begin{align*}
 \min \, \{ z_1,\ldots, z_N \} - \frac{\log N }{\rho} 
 &\le \mbox{softmin}_\rho(z_1,\ldots,z_N) \\
 &\le \min \, \{z_1,\ldots, z_N\}.
\end{align*}
}
\end{proposition}
Proposition~\ref{prop:softmin_bound} demonstrates that $\mbox{softmin}_\rho$ lower bounds the minimum, and converges to the minimum as $\rho \to \infty$. 
Thus, $\mbox{softmin}_\rho$ is a smooth approximation of the minimum. 
Note that if $N>1$, then the soft minimum is strictly less than the minimum.

\section{Problem Formulation} \label{section:problem_formulation}
Consider 
\begin{equation}\label{eq:dynamics}
\dot x(t) = f(x(t))+g(x(t)) u(x(t)),
\end{equation}
where $x(t) \in \BBR^{n}$ is the state, $x(0) = x_0 \in \BBR^n$ is the initial condition, $f:\BBR^n \to \BBR^n$ and $g:\BBR^n\to \BBR^{n\times m}$ are locally Lipschitz continuous on $\BBR^n$, and $u: \BBR^n \to \BBR^m$ is the control, which is locally Lipschitz continuous on $\BBR^n$.

Since $f$, $g$, and $u$ are locally Lipschitz, it follows that for all $x_0 \in \BBR^n$, there exists a maximum value $t_{\rm max}(x_0) \in [0,\infty)$ such that $x(t)$ is the unique solution to~\eqref{eq:dynamics} on $I(x_0) \triangleq [0, t_{\rm max}(x_0))$.

\subsection{Preliminary Definitions and Results}

The following definitions are needed. 
The first definition is similar to~\cite[Definition 4.4]{blanchini2008set}. 
The second definition is a standard definition of a CBF (for example, see~\cite[Definition 5]{ames2016control}), and the third definition is a relaxation of the standard CBF definition.

\begin{definition}\label{def:cont_fwd_inv}
\rm
A set $\SD \subset \BBR^n$ is \textit{control forward invariant} with respect to~\eqref{eq:dynamics} if there exists a locally Lipschitz $u_\rmi \colon \SD \to \BBR^m$ such that for all $x_0\in\SD$, the solution $x$ to \eqref{eq:dynamics} with $u=u_\rmi$ is such that for all $t\in I(x_0)$, $x(t) \in \SD$.
\end{definition}

\begin{definition}\label{def:CBF}
\rm
Let $\eta \colon \BBR^n \to \BBR$ be continuously differentiable, and define $\SD \triangleq \{ x \in \BBR^n \colon \eta(x) \ge 0 \}$.
Then, $\eta$ is a \textit{control barrier function} (CBF) for~\cref{eq:dynamics} on $\SD$ if there exists an extended class-$\SK$ function $a \colon \BBR \to \BBR$ such that for all $x \in \SD$,
    \begin{equation}
        \sup_{u\in \BBR^{m}} L_f \eta(x) + L_g \eta(x)u +a(\eta(x))\ge 0.\label{def:CBF.1}
    \end{equation} 
\end{definition}

\begin{definition}\label{def:RCBF}
\rm
Let $\eta \colon \BBR^n \to \BBR$ be continuously differentiable, and define $\SD \triangleq \{ x \in \BBR^n \colon \eta(x) \ge 0 \}$.
Then, $\eta$ is a \textit{relaxed control barrier function} (R-CBF) for~\cref{eq:dynamics} on $\SD$ if for all $x \in \mbox{bd }\SD$,
    \begin{equation}
        \sup_{u\in \BBR^{m}} L_f \eta(x) + L_g \eta(x)u \ge 0.\label{def:RCBF.1}
    \end{equation} 
\end{definition}

\Cref{def:RCBF} implies that an R-CBF need only satisfy \eqref{def:RCBF.1} on the boundary of its zero-superlevel.
In contrast, \Cref{def:CBF} implies that a CBF must satisfy \eqref{def:CBF.1} on the boundary as well as the interior of its zero-superlevel set.
Note that every CBF is an R-CBF. 
CBF definitions often take the supremum \eqref{def:CBF.1} over a subset of $\BBR^m$ based on control limits. 
Nevertheless, \Cref{def:CBF,def:RCBF} are adequate for this article because we do not require CBFs or R-CBFs where the input variable is constrained to a subset of $\BBR^m$ even though 
Section~\ref{section:input constraints} addresses control constraints.

The following result provides sufficient conditions for $\eta$ to be an R-CBF.
The result also provides sufficient conditions for the zero-superlevel set of $\eta$ to be control forward invariant.

\begin{lemma}\label{lem:rcbf}
\rm 
Let $\eta \colon \BBR^n \to \BBR$ be continuously differentiable, and define $\SD \triangleq \{ x \in \BBR^n \colon \eta(x) \ge 0 \}$. 
Assume that for all $x \in \mbox{bd }\SD$, if $L_f \eta(x) \le 0$, then $L_g \eta(x) \ne 0$. 
Then, the following statements hold:
\begin{enumalph}
\item \label{lem:rcbf.1}
$\eta$ is an R-CBF for~\eqref{eq:dynamics} on $\SD$. 
\item \label{lem:rcbf.2}
If $\eta^\prime$ is locally Lipschitz on $\SD$, then $\SD$ is control forward invariant with respect to~\eqref{eq:dynamics}.
\end{enumalph}
\end{lemma}

\begin{proof}

To prove \ref{lem:rcbf.1}, consider $u_\rmi \colon \SD \to \BBR^m$ defined by 
\begin{equation}\label{eq:u_i}
u_\rmi(x) \triangleq
\begin{cases}
    \frac{-L_f \eta(x) L_g \eta(x)^\rmT}{L_g \eta(x)L_g \eta(x)^\rmT+\eta(x)^2}, & L_f \eta(x) < 0,\\
     0, & L_f \eta(x) \ge 0.\\    
\end{cases}
\end{equation}
and it follows that for all $x \in \mbox{bd } \SD$, 
\begin{equation}\label{eq:eta_dot}
L_f \eta(x) + L_g \eta(x)u_\rmi(x) = 
\begin{cases}
    0, & L_f \eta(x) < 0,\\
     L_f \eta(x), & L_f \eta(x) \ge 0.\\    
\end{cases}
\end{equation}
Thus, \eqref{eq:eta_dot} implies that for all $x \in \mbox{bd } \SD$, $L_f \eta(x) + L_g \eta(x)u_\rmi(x) \ge 0$, and it follows from \Cref{def:RCBF} that $\eta$ is an R-CBF.

To prove \ref{lem:rcbf.2}, since $L_g \eta(x) \ne 0$ for all $x \in \{ a \in \mbox{bd } \SD \colon L_f \eta(a) \le 0 \}$, it follows from \eqref{eq:u_i} that $u_\rmi$ is continuous on $\SD$. 
Since, in addition, $f$, $g$, and $\eta^\prime$ are locally Lipschitz on $\SD$, it follows from \eqref{eq:u_i} that $u_\rmi$ is locally Lipschitz on $\SD$. 
Next, consider~\eqref{eq:dynamics} with $u = u_\rmi$ and $x_0 \in \SD$. 
Since for all $x \in \mbox{bd } \SD$, $L_f \eta(x) + L_g \eta(x)u_\rmi(x) \ge 0$, it follows from Nagumo's Theorem~\cite[Theorem 4.7]{blanchini2008set} that for all $t \in I(x_0)$, $x(t) \in \SD$.
Thus, \Cref{def:cont_fwd_inv} implies that $\SD$ is control forward invariant.
\end{proof}

\subsection{Control Objective}

Let  $h_1,h_2,\ldots,h_\ell \colon \BBR^n \to \BBR$ be continuously differentiable, and for all $j \in \{ 1,2,\ldots,\ell\}$, define 
\begin{equation}
    \SC_{j,0} \triangleq \{ x \in \BBR^n \colon h_j(x)\geq 0 \}. 
\end{equation}
The \textit{safe set} is
\begin{equation}\label{eq:S_s}
\SSS_\rms \triangleq {\bigcap_{j=1}^\ell} \SC_{j,0}.
\end{equation}
Unless otherwise stated, all statements in this paper that involve the subscript $j$ are for all $j \in \{ 1,2,\ldots,\ell\}$.
We make the following assumption:
\begin{enumA}

\item There exists a positive integer $d_j$ such that for all $x \in \BBR^n$ and all $i \in \{ 0,1,\ldots,d_j-2\}$, $L_g L_f^ih_j(x) = 0$; and for all $x \in \SSS_\rms$, $L_g L_f^{d_j-1}h_j(x) \neq 0$. \label{cond:hocbf.a}

\end{enumA}

Assumption~\ref{cond:hocbf.a} implies $h_j$ has well-defined relative degree $d_j$ with respect to~\eqref{eq:dynamics} on $\SSS_\rms$; however, relative degrees $d_1,\ldots,d_\ell$ need not be equal.
Assumption~\ref{cond:hocbf.a} also implies that $h_j$ is a relative-degree-$d_j$ CBF. 
However, we do not assume knowledge of a CBF for the safe set $\SSS_\rms$.
Section~\ref{section:composite CBF} presents a method for constructing a single composite R-CBF from $h_1,\ldots,h_\ell$, which can have different relative degrees.

Consider the cost function $J \colon \BBR^n \times \BBR^m\to \BBR$ defined by
\begin{equation}\label{eq:quad_cost}
    J(x, u) \triangleq \frac{1}{2}u^\rmT Q(x) u + c(x)^\rmT u,
\end{equation}
where $Q:\BBR^n \to \BBP^{m}$ and $c:\BBR^n\to \BBR^m$ are locally Lipschitz continuous on $\BBR^n$.

The control objective is to design a full-state feedback control $u:\BBR^n \to \BBR^m$ such that for all $t \in I(x_0)$, $J(x(t), u(x(t))$ is minimized subject to the safety constraint that $x(t) \in \SSS_\rms$.
Section~\ref{section:control} presents a closed-form control that satisfies these objectives. 
Then, Section~\ref{section:input constraints} presents a closed-form control that satisfies these objectives and satisfies control input constraints.

A special case of this optimal control problem is the minimum-intervention problem, where we consider a desired control $u_\rmd : \BBR^n \to \BBR^m$ that is designed to satisfy performance requirements but may not account for safety. 
In this case, the objective is to design a feedback control such that the minimum-intervention cost $\| u - u_\rmd(x(t)) \|_2^2$ is minimized subject to the safety constraint and potentially subject to control input constraints.  
The minimum-intervention cost $\| u - u_\rmd(x) \|_2^2$ is minimized by $u_\rmd(x)$, which is equal to the minimizer of \eqref{eq:quad_cost} with $Q(x) = I_m$ and $c(x) = -u_\rmd(x)$. 
Thus, the minimum-intervention problem is addressed by letting $Q(x) = I_m$ and $c(x) = -u_\rmd(x)$.

Although the cost \eqref{eq:quad_cost} is quadratic in $u$, nonquadratic cost functions can often be approximated locally by quadratic forms. 
This quadratic approximation technique is commonly used in differential dynamic programming (e.g., \cite{tassa2014control}). 
Thus, the method in this paper can be implemented effectively in some cases where the cost is not quadratic in $u$.

\section{Composite Soft-Minimum R-CBF}
\label{section:composite CBF}

This section presents a method for constructing a single composite R-CBF from multiple CBFs (i.e., $h_1,\ldots,h_\ell$), which can have different relative degrees.

Let $b_{j,0}(x) \triangleq h_j(x)$. 
For $i\in\{0, 1\ldots,d_j-2\}$, let $\alpha_{j,i} \colon \BBR \to \BBR$ be a locally Lipschitz extended class-$\SK$ function, and consider $b_{j, i+1}:\BBR^n \to \BBR$ defined by
\begin{equation}\label{eq:b_def}
    b_{j,i+1}(x) \triangleq L_f b_{j,i}(x) + \alpha_{j,i}(b_{j,i}(x)).
\end{equation}
For $i \in \{1\ldots,d_j-1\}$, define
\begin{equation}\label{eq:C_j_i_def}
    \SC_{j,i} \triangleq \{x\in \BBR^n: b_{j,i}(x) \ge 0\}.
\end{equation}
Next, define
\begin{gather}  \label{eq:C_j_def}
    \SC_j \triangleq 
        \begin{cases}
        \SC_{j,0}, & d_j = 1,\\
    \bigcap_{i=0}^{d_j - 2} \SC_{j,i}, & d_j > 1,
    \end{cases}
\end{gather}
and 
\begin{equation}  
    \SC \triangleq \bigcap_{j=1}^{\ell} \SC_j. \label{eq:C_def}
\end{equation}
Note that $\SC \subseteq \SSS_\rms$.
In addition, note that if $d_1,\ldots,d_\ell \in \{1,2\}$, then $\SC = \SSS_\rms$.

The next result is from \cite[Proposition 1]{tan2021high} and provides a sufficient condition such that $\SC_j$ is forward invariant.

\begin{lemma}\label{lemma:fwd_inv}
\rm
Consider~\eqref{eq:dynamics}, where \ref{cond:hocbf.a} is  satisfied.
Let $j \in\{1,\ldots,\ell\}$.
Assume $x_0 \in \SC_{j}$, and assume for all $t\in I(x_0)$, $b_{j, d_j-1}(x(t))\ge 0$. 
Then, for all $t \in I(x_0)$, $x(t) \in \SC_{j}$.
\end{lemma}

Lemma~\ref{lemma:fwd_inv} implies that if $x_0 \in \SC$ and for all $j \in\{1,\ldots,\ell\}$ and all $t\in I(x_0)$, $b_{j, d_j-1}(x(t))\ge 0$, then for all $t\in I(x_0)$, $x(t) \in \SC \subseteq \SSS_\rms$.
This motivates us to consider a candidate R-CBF whose zero-superlevel set approximates the intersection of the zero-superlevel sets of $b_{1, d_1-1},\ldots,b_{\ell, d_\ell-1}$. 
Specifically, let $\rho > 0$, and consider the candidate R-CBF $h:\BBR^n \to \BBR$ defined by
\begin{equation}\label{eq:h_def}
    h(x) \triangleq \mbox{softmin}_\rho \Big ( b_{1, d_1-1}(x), b_{2, d_2-1}(x), \ldots,b_{\ell, d_\ell-1}(x) \Big ). 
\end{equation}
The zero-superlevel set of $h$ is
\begin{equation}\label{eq:defS_softmin}
\SH \triangleq \{ x \in \BBR^n \colon h(x) \ge 0 \}.
\end{equation}
The next result is the immediate consequence of Proposition~\ref{prop:softmin_bound} and demonstrates that $\SH$ is a subset of the intersection of the zero-superlevel sets of $b_{1, d_1-1},\ldots, b_{\ell, d_\ell-1}$.

\begin{proposition}\label{prop:S}
$\SH \subseteq \bigcap_{j=1}^{\ell} \SC_{j,d_j-1}$.
\end{proposition}

Proposition~\ref{prop:softmin_bound} also implies that $\SH$ approximates the intersection of the zero-superlevel sets of $b_{1, d_1-1},\ldots, b_{\ell, d_\ell-1}$ in the sense that as $\rho \to \infty$, $\SH \to \bigcap_{j=1}^{\ell} \SC_{j,d_j-1}$. 
In other words, Proposition~\ref{prop:softmin_bound} shows that for sufficiently large $\rho>0$, $h$ is a smooth approximation of $\min \, \{ b_{1, d_1-1}, \ldots,b_{\ell, d_\ell-1}\}$. 
If $\rho >0$ is small, then $h$ is a conservative approximation of $\min \, \{ b_{1, d_1-1}, \ldots,b_{\ell, d_\ell-1}\}$. 
However, if $\rho >0$ is large, then $\| h^\prime(x) \|_2$ is large at points where $\min \, \{ b_{1, d_1-1}, \ldots,b_{\ell, d_\ell-1}\}$ is not differentiable.
Thus, selecting $\rho$ is a trade-off between the conservativeness of $h$ and the size of $\| h^\prime(x) \|_2$.

\begin{figure*}[ht!]
\includegraphics[width=\textwidth,clip=true,trim= 0.0in 0.0in 0in 0.0in] {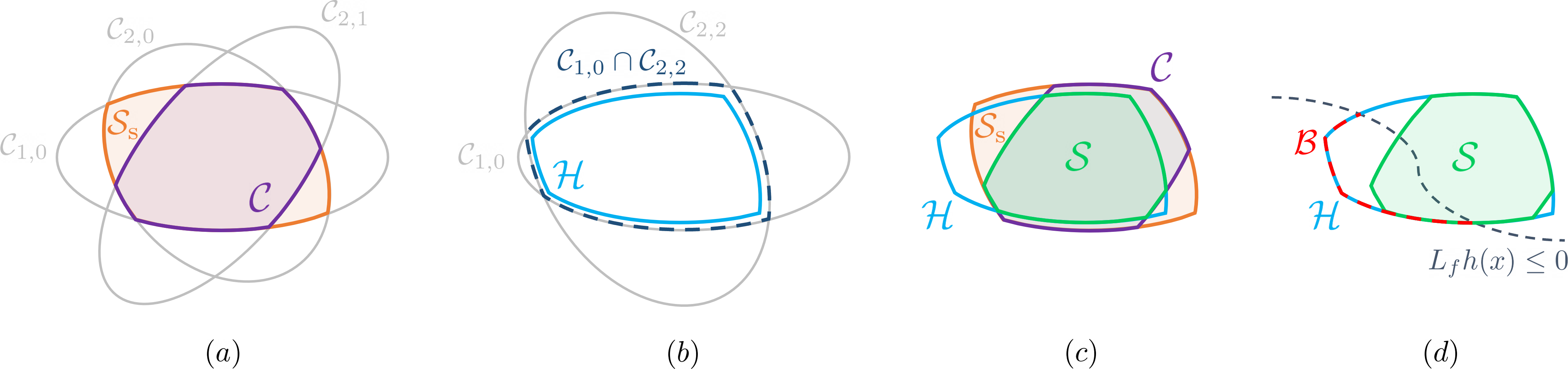}
\caption{Illustration of the relationships between $\SSS_\rms$, $\SC$, $\SH$, and $\SSS$ for $\ell = 2$ with $d_1 = 1$ and $d_2 = 3$. 
(a) shows $\SSS_\rms = \SC_{1,0} \cap \SC_{2,0}$ and $\SC = \SC_1 \cap \SC_2 = \SC_{1,0}\cap \SC_{2,0} \cap \SC_{2,1}$. 
(b) shows $\SC_{1,0} \cap \SC_{2,2}$ and $\SH\subset \SC_{1,0} \cap \SC_{2,2}$. 
(c) shows $\SSS_\rms$, $\SC$, $\SH$, and $\SSS \triangleq \SC \cap \SH$. 
(d) shows $\SH$, $\SSS$ and $\SB \triangleq \{ x \in \bd \SH \colon L_fh(x) \le 0 \}$. 
} \label{fig:set_diag}
\end{figure*}

Note that $\SH$ is not generally a subset of $\SSS_\rms$ or $\SC$. 
In the special case where $d_1=\cdots=d_\ell = 1$, it follows that $\SH \subseteq \SC = \SSS_\rms$, and as $\rho \to \infty$, $\SH \to \SSS_\rms$.

Next, we define 
\begin{equation}\label{eq:S}
    \SSS \triangleq \SH \cap \SC,
\end{equation}
and since $\SC \subseteq \SSS_\rms$, it follows that $\SSS \subseteq \SSS_\rms$.

To illustrate the relationship between $\SSS_\rms$, $\SC$, $\SH$, and $\SSS$, we consider the case where $\ell =2$, $d_1=1$, and $d_2 = 3$. 
In this case, $b_{1,0} = h_1$, $b_{2,0} = h_2$, and $b_{2,1}$ and $b_{2,2}$ are defined by~\cref{eq:b_def}.
First, it follows from~\cref{eq:S_s} that $\SSS_\rms = \SC_{1,0}\cap \SC_{2,0}$. 
Next, it follows from~\cref{eq:C_j_i_def,eq:C_j_def,eq:C_def} that $\SC = \SC_1 \cap \SC_2 = \SC_{1,0} \cap \SC_{2,0} \cap \SC_{2,1}$.
Thus, $\SC \subset \SSS_\rms$ as shown in \Cref{fig:set_diag}(\textit{a}).
For this example, \cref{eq:h_def} implies that $h(x) = \mbox{softmin}_\rho(b_{1,0}(x), b_{2,2}(x))$, and it follows from \Cref{prop:S} that $\SH \subset \SC_{1,0} \cap \SC_{2,2}$ as shown in \Cref{fig:set_diag}(\textit{b}).
Finally, as shown in \Cref{fig:set_diag}(\textit{c}), $\SH$ is not necessarily a subset of $\SSS_\rms$, but $\SSS = \SH \cap \SC$ is a subset of $\SSS_\rms$.

Next, define
\begin{align}
    \SB &\triangleq \{ x \in \bd \SH \colon L_fh(x) \le 0 \}\nn\\
    &=\{ x \in \BBR^n \colon L_fh(x) \le 0 \mbox{ and } h(x) =0 \}, \label{eq:SB}
\end{align}
which is shown in \Cref{fig:set_diag}(\textit{d}).
\Cref{lem:rcbf} implies that if $L_gh$ is nonzero on $\SB$, then $h$ is an R-CBF.
\Cref{lem:rcbf} also implies that if $L_gh$ is nonzero on $\SB$ and $h^\prime$ is locally Lipschitz, then $\SH$ is control forward invariant.
The next result shows that in this case not only is $\SH$ control forward invariant but $\SSS$ is also control forward invariant.
This fact is significant because $\SSS$ is a subset of the safe set $\SSS_\rms$, whereas $\SH$ is not necessarily a subset of $\SSS_\rms$.

\begin{proposition}\label{prop:cont_fwd_S_C}
\rm{
Consider~\eqref{eq:dynamics}, where \ref{cond:hocbf.a} is satisfied. 
Assume that $h^\prime$ is locally Lipschitz on $\SH$, and for all $x \in \SB$, $L_gh(x) \ne 0$.
Then, $\SSS$ is control forward invariant.
}
\end{proposition}

\begin{proof}

Let $x_0 \in \SSS$. 
Since $h^\prime$ is locally Lipschitz on $\SH$, and for all $x \in \SB$, $L_gh(x) \ne 0$, it follows from \Cref{lem:rcbf} that $\SH$ is control forward invariant. 
Since, in addition, $x_0 \in \SSS \subseteq \SH$, it follows from \Cref{def:cont_fwd_inv} that there exists a locally Lipschitz $u_\rmi \colon \SH \to \BBR^m$ such that the solution to \eqref{eq:dynamics} with $u=u_\rmi$ is such that for all $t\in I(x_0)$, $x(t) \in \SH$, which implies $h(x(t)) \ge 0$.
Thus, \eqref{eq:defS_softmin}, Proposition~\ref{prop:S}, and \eqref{eq:C_j_i_def} imply that for all $t \in I(x_0)$, $b_{j, d_j-1}(x(t)) \ge 0$.
Since, in addition, $x_0 \in \SSS \subseteq \SC_j$, it follows from Lemma~\ref{lemma:fwd_inv} that for all $t \in I(x_0)$, $x(t) \in \SC_j$. 
Thus, for all $t \in I(x_0)$, $x(t) \in \SC$, which implies that for all $t \in I(x_0)$, $x(t) \in \SSS = \SH \cap \SC$.
\end{proof}

\begin{remark}\label{rem:cont_fwd_S_C}\rm
Proposition~\ref{prop:cont_fwd_S_C} provides a sufficient condition such that $\SSS$ is control forward invariant.
However, \Cref{fig:set_diag}(\textit{d}) illustrates that it is not necessary that $L_gh(x) \neq 0$ for all $x \in \SB$.
Specifically, it suffices that for all $x \in \SB \cap \SSS$, $L_gh(x) \ne 0$.
\end{remark}

The conditions in Proposition~\ref{prop:cont_fwd_S_C} guarantee that $\SSS \subseteq \SSS_\rms$ is control forward invariant and $h$ is an R-CBF. 
In this case, $h$ is not necessarily a CBF. 
Under additional assumptions (e.g., $\SH$ is compact), $h$ may be a CBF. 
However, even if $h$ is a CBF, it can be difficult to determine an associated extended class-$\SK$ function $a \colon \BBR \to \BBR$ that satisfies \Cref{def:CBF}. 
The next section presents an optimal and safe control using the R-CBF $h$.
This approach does not require that $h$ is CBF, and thus, does not require knowledge of an associated function~$a$.

\section{Closed-Form Optimal and Safe Control} \label{section:control}

This section uses the composite soft-minimum R-CBF \eqref{eq:h_def} to construct a closed-form optimal control that guarantees safety. 
Specifically, we design an optimal control subject to the constraint that $x(t) \in \SSS \subseteq \SSS_\rms$.

Let $\gamma > 0$, and let $\alpha:\BBR \to \BBR$ be a locally Lipschitz nondecreasing function such that $\alpha(0) = 0$. 
For all $x \in \BBR^n$, consider the control given by
\begin{subequations}\label{eq:qp_softmin}
\begin{align}
& \Big ( u(x), \mu(x) \Big ) \triangleq \underset{{\tilde u\in \BBR^m, \, \tilde \mu \in \BBR} }{\mbox{argmin}}  \, 
J(x, \tilde u) + \frac{1}{2}\gamma \tilde \mu^2 \label{eq:qp_softmin.a}\\
& \text{subject to}\nn\\
& L_f h(x) + L_g h(x) \tilde u + \alpha (h(x)) + \tilde \mu h(x) \ge 0. \label{eq:qp_softmin.b}
\end{align}
\end{subequations}
The quadratic program \eqref{eq:qp_softmin} includes the slack variable $\tilde \mu$ because $h$ is an R-CBF but not necessarily a CBF.
A typical CBF-based constraint does not include the slack-variable term $\tilde \mu h(x)$ in \eqref{eq:qp_softmin.b}.
In this case, additional assumptions are needed to ensure that 
\eqref{eq:qp_softmin.b} is feasible. 
For example, it is common to require that $h$ is a CBF and that $\alpha$ is an extended class-$\SK$ function that satisfies \Cref{def:CBF} with $\eta = h$ and $a = \alpha$.  
Since $h$ is only an R-CBF, the slack variable $\tilde \mu$ ensures that for all $x \not \in \SB$, \eqref{eq:qp_softmin.b} is feasible.
Hence, if $L_gh$ is nonzero on $\SB$, then \eqref{eq:qp_softmin.b} is feasible for all $x \in \BBR^n$. 
The next result formalizes this fact.
Specifically, it shows that if for all $x \in \SB$, $L_gh(x) \ne 0$, then the quadratic program~\eqref{eq:qp_softmin} is feasible on $\BBR^n$.

\begin{proposition}\label{lemma:qp_feas}\rm
The following statements hold:
\begin{enumalph}

\item \label{lemma:qp_feas.a}
Let $x \in \BBR^n \setminus \SH$. 
Then, there exists $\tilde u \in \BBR^m$ and $\tilde \mu \in \BBR$ such that~\eqref{eq:qp_softmin.b} is satisfied.

\item \label{lemma:qp_feas.a2}
Let $x \in \SH \setminus \SB$.
Then, there exists $\tilde u \in \BBR^m$ and $\tilde \mu \ge 0$ such that~\eqref{eq:qp_softmin.b} is satisfied.

\item \label{lemma:qp_feas.b}
Let $x \in \SB$. 
If $L_g h(x) \ne 0$, then there exists $\tilde u \in \BBR^m$ and $\tilde \mu \ge 0$ such that~\eqref{eq:qp_softmin.b} is satisfied. 
\end{enumalph}
\end{proposition}

\begin{proof}

To prove \ref{lemma:qp_feas.a}, let $x_1 \in \BBR^n \setminus \SH$, which implies that $h(x_1) < 0$.
Thus, $\tilde u = 0$ and $\tilde \mu = - [L_f h(x_1) + \alpha(h(x_1))]/h(x_1)$ satisfy~\eqref{eq:qp_softmin.b}.

To prove \ref{lemma:qp_feas.a2}, let $x_2\in (\bd \SH )\setminus \SB$, which implies that $h(x_2) = 0$ and $L_fh(x_2) > 0$.
Thus, $\tilde u = 0$ and $\tilde \mu = 0$ satisfy~\eqref{eq:qp_softmin.b}.
Next, let $x_3 \in \mbox{int }\SH$, which implies that $h(x_3) > 0$. 
Thus, $\tilde u = 0$ and $\tilde \mu = \max\{- [L_f h(x_3) + \alpha(h(x_3))]/{h(x_3)},0\} \ge 0$ satisfy~\eqref{eq:qp_softmin.b}.

To prove \ref{lemma:qp_feas.b}, let $x_4\in \SB$, which implies that $h(x_4) = 0$. 
Since $L_gh(x_4) \neq 0$, it follows that $\tilde u = -L_fh(x_4)/L_gh(x_4)$ and $\tilde \mu =0$ satisfy~\eqref{eq:qp_softmin.b}. 
\end{proof}

The user-selected parameter $\gamma >0$ influences the behavior of the control $u$ that satisfies \eqref{eq:qp_softmin}.
Specifically, $\gamma$ weights the slack-variable term in the cost $J(x, \tilde u) + \frac{1}{2}\gamma \tilde \mu^2$.
For small $\gamma$ (e.g., as $\gamma \to 0$), the control $u$ that satisfies \eqref{eq:qp_softmin} is aggressive in the sense that it is approximately equal to $-Q^{-1}c$, which is the unconstrained minimizer of $J$, except for $x$ near the boundary of $\SH$.
For large $\gamma$ (e.g., as $\gamma \to \infty$), the control $u$ that satisfies \eqref{eq:qp_softmin} is similar to the control generated by a quadratic program that does not include the slack-variable term in \eqref{eq:qp_softmin.b}.
In this case, the aggressiveness of the control is determined by the nondecreasing function $\alpha$ and the extended class-$\SK$ functions $\alpha_{j,0},\ldots,\alpha_{j,d_j-2}$.

In order to present a closed-form solution to the quadratic program \eqref{eq:qp_softmin}, consider $\omega:\BBR^n \to \BBR$ defined by
\begin{equation}\label{eq:omega}
    \omega(x) \triangleq L_fh(x) - L_gh(x)Q(x)^{-1} c(x) +\alpha(h(x)),
\end{equation}
and define
\begin{equation}\label{eq:Omega}
    \Omega \triangleq \{x\in\BBR^n: \omega(x) < 0\}.
\end{equation}
The following result provides a closed-form solution for the unique global minimizer $( u(x), \mu(x) )$ of the constrained optimization \eqref{eq:qp_softmin}. 
This result also shows that if $h^\prime$ is locally Lipschitz, then $u$ and $\mu$ are locally Lipschitz.

\begin{theorem}\label{thm:lip}\rm
Assume that for all $x\in\SB$, $L_gh(x) \neq 0$. 
Then, the following hold:
\begin{enumalph}

\item \label{thm:lip.a}
For all $x \in \BBR^n$,
\begin{align}
u(x) &= - Q(x)^{-1} \left(c(x) - L_gh(x)^\rmT \lambda(x)\right) ,\label{eq:u_sol}\\
\mu(x) &= \frac{h(x) \lambda(x)}{\gamma},\label{eq:mu_sol}
\end{align}
where $\lambda \colon \BBR^n \to \BBR$ is defined by
\begin{equation}\label{eq:lambda_sol}
\lambda(x) \triangleq
    \begin{cases}
    \frac{-\omega(x)}{d(x)}, & x\in\Omega,\\
     0, & x\notin\Omega,\\    
    \end{cases}
\end{equation}
and $d \colon \Omega \to \BBR$ is defined by 
\begin{equation}\label{eq:den}
    d(x) \triangleq L_gh(x)Q(x)^{-1} L_gh(x)^\rmT + \gamma^{-1} h(x)^2.
\end{equation}

\item \label{thm:lip.b}
For all $x \in \BBR^n$, $\lambda(x)\ge 0$, and for all $x \in \SH$, $\mu(x) \ge 0$.

\item 
$u$, $\mu$, and $\lambda$ are continuous of $\BBR^n$. \label{thm:lip.c}

\item 
Let $\SD \subseteq \BBR^n$, and assume that $h^\prime$ is locally Lipschitz on $\SD$.
Then, $u$, $\mu$, and $\lambda$ are locally Lipschitz on $\SD$. \label{thm:lip.d}
\end{enumalph}
\end{theorem}

\begin{proof}

First, we show that for all $x \in \mbox{cl }\Omega$, $d(x) > 0$.
Let $a \in \mbox{cl }\Omega$, and assume for contradiction that $d(a) = 0$. 
Since $\gamma > 0$ and $Q$ is positive definite, it follows from \eqref{eq:den} that $L_g h(a) = 0$ and $h(a)=0$.
Since, in addition, for all $x \in \SB$, $L_gh(x) \neq 0$, it follows from \eqref{eq:SB} that $L_fh(a) > 0$. 
Thus, \eqref{eq:omega} implies $\omega(a) = L_f h(a) > 0$, which implies $a \notin \mbox{cl }\Omega$, which is a contradiction. 
Thus, $d(a) \not = 0$, which implies that $d(a) > 0$. 
Thus, for all $x \in \mbox{cl }\Omega$, $d(x) > 0$.

To prove \ref{thm:lip.a}, define 
\begin{gather*}
\tilde J(x, \tilde u, \tilde \mu) \triangleq J(x, \tilde u) + \frac{1}{2}\gamma \tilde \mu^2,\\
b(x, \tilde u, \tilde \mu) \triangleq L_fh(x) + L_gh(x) \tilde u + \alpha(h(x)) + \tilde \mu h(x),\\
u_*(x) \triangleq -Q(x)^{-1} c(x),
\end{gather*}
and note that $u_*$ is the unique global minimizer of $J$, which implies that $(u_*,0)$ is the unique global minimizer of $\tilde J$.

First, let $x_1 \notin \Omega$, and it follows from \eqref{eq:Omega} that $\omega(x_1) \ge 0$, which combined with \eqref{eq:omega} implies that $(\tilde u,\tilde \mu) = (u_*(x_1),0)$ satisfies~\eqref{eq:qp_softmin.b}.
Since, in addition, $(\tilde u,\tilde \mu) = (u_*(x_1),0)$ is the unique global minimizer of $\tilde J(x_1,\tilde u, \tilde \mu)$, it follows that $(\tilde u,\tilde \mu) = (u_*(x_1),0)$ is the solution to~\eqref{eq:qp_softmin}. 
Finally, \Cref{eq:u_sol,eq:mu_sol,eq:lambda_sol,eq:den} yields $u(x_1) = u_*(x_1)$ and $\mu(x_1) = 0$, which confirms \ref{thm:lip.a} for all $x \notin \Omega$.

Next, let $x_2 \in \Omega$. 
Let $(u_2, \mu_2) \in \BBR^{m} \times \BBR$ denote the the unique global minimizer of $\tilde J(x_2, \tilde u, \tilde \mu)$ subject to $b(x_2, \tilde u, \tilde \mu) \ge 0$. 
Since $x_2 \in \Omega$, it follows from \eqref{eq:omega} that $b(x_2, u_*(x_2), 0) = \omega(x_2) < 0$.
Thus, $b(x_2, u_2, \mu_2) =0$.
Define the Lagrangian
\begin{equation*}
\SL(\tilde u, \tilde \mu, \tilde \lambda) \triangleq \tilde J(x_2, \tilde u , \tilde \mu) - \tilde \lambda b(x_2, \tilde u, \tilde \mu).
\end{equation*}
Let $\lambda_2 \in \BBR$ be such that $(u_2, \mu_2,\lambda_2)$ is a stationary point of $\SL$. 
Evaluating $\pderiv{\SL}{{\tilde u}}$, $\pderiv{\SL}{\tilde \mu}$, and $\pderiv{\SL}{\tilde \lambda}$ at $(u_2, \mu_2,\lambda_2)$; setting equal to zero; and solving for $u_2$, $\mu_2$, and $\lambda_2$ yields 
\begin{gather*}
u_2 = -Q(x_2)^{-1}\left(c(x_2) - L_g h(x_2)^\rmT \lambda_2\right),\\
\mu_2 = \frac{h(x_2)\lambda_2}{\gamma}, \\ 
\lambda_2 = \frac{-\omega(x_2)}{d(x_2)},
\end{gather*}
where $d(x_2) \neq 0$ because $x_2 \in \Omega$. 
Finally,~\Cref{eq:u_sol,eq:mu_sol,eq:lambda_sol,eq:den} yields $u(x_2) = u_2$, $\mu(x_2) = \mu_2$, and $\lambda(x_2) = \lambda_2$, which confirms \ref{thm:lip.a} for all $x \in \Omega$.

To prove \ref{thm:lip.b}, since $d$ is positive on $\mbox{cl } \Omega$, it follows from~\cref{eq:omega,eq:Omega,eq:lambda_sol} that for all $x\in\BBR^n$, $\lambda(x) \ge 0$. 
Since, in addition, $\gamma > 0$ and $h$ is nonnegative on $\SSS$, it follows from \cref{eq:mu_sol} that for all $x\in\SH$, $\mu(x) \ge 0$.

To prove~\ref{thm:lip.c}, let $a \in \bd \Omega$, which implies that $\omega(a) =0$ and $d(a) > 0$. 
Thus, $-\omega(a)/d(a) =0$, and \eqref{eq:lambda_sol} implies that $\lambda$ is continuous on $\bd \Omega$. 
Since, in addition, $f$, $g$, $Q^{-1}$, $c$, $h$, and $h^\prime$ are continuous on $\BBR^n$, it follows from \cref{eq:omega,eq:lambda_sol,eq:den} that $\lambda$ is continuous on $\BBR^n$, which combined with \eqref{eq:u_sol} and \eqref{eq:mu_sol} implies that $u$ and $\mu$ are continuous on $\BBR^n$.

To prove~\ref{thm:lip.d}, note that $f$, $g$, $Q^{-1}$, $c$, $\alpha$ and $h$ are locally Lipschitz on $\BBR^n$. 
Since $h^\prime$ is locally Lipschitz on $\SD$, it follows from \cref{eq:omega,eq:den} that $\omega$ and $d$ are locally Lipschitz on $\SD$. 
Thus, \cref{eq:lambda_sol} implies that $\lambda$ is locally Lipschitz on $\SD$, which combined with \cref{eq:u_sol,eq:mu_sol} implies that $u$ and $\mu$ are locally Lipschitz on $\SD$.
\end{proof}

\begin{figure}[t!]
\center{\includegraphics[width=0.5\textwidth,clip=true,trim= 0.0in 0.0in 0in 0.0in] {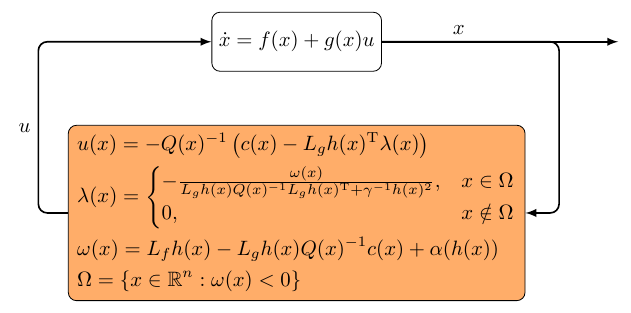}}
\caption{Closed-form optimal and safe control using the composite soft-minimum R-CBF \eqref{eq:h_def}.
Control minimizes cost subject to safety constraint.} 
\label{fig:block_diag1}
\end{figure}

Figure~\ref{fig:block_diag1} is a block diagram of the control~\Cref{eq:b_def,eq:h_def,eq:u_sol,eq:omega,eq:Omega,eq:mu_sol,eq:lambda_sol,eq:den}, which is the closed-form solution to the quadratic program~\eqref{eq:qp_softmin}. 
The control depends on the user-selected parameter $\gamma >0$ and the user-selected functions $\alpha$ and $\alpha_{j,1},\ldots,\alpha_{j,d_j-2}$, which are similar to the extended class-$\SK$ functions used in a typical higher-order CBF-based constraint. 
The parameter $\gamma$ and the function $\alpha$ influence the aggressiveness of the control.
The functions $\alpha_{j,1},\ldots,\alpha_{j,d_j-2}$ influence the aggressiveness of the control as well as the size of $\SC$ and thus, the size of the control forward invariant set $\SSS = \SH \cap \SC$.

The control~\Cref{eq:b_def,eq:h_def,eq:u_sol,eq:omega,eq:Omega,eq:mu_sol,eq:lambda_sol,eq:den} requires computation of $Q(x)^{-1}$.
For many applications, the cost $J$ is known in advance of online implementation. 
In this case, $Q(x)^{-1}$ can be computed offline and in closed form. 
For example, consider the minimum-intervention specialization discussed in \Cref{section:problem_formulation}.
In this case, $Q = Q^{-1} = I_m$. 
If an application requires online computation of $Q(x)^{-1}$, then this typically has a worst-case complexity of $\SO(m^3)$.
For comparison, solving the quadratic program~\eqref{eq:qp_softmin} using an interior-point method with an $\epsilon$-tolerance typically has a complexity of $\SO(m^3 \log(1/\epsilon))$
(e.g., \cite[Chapter 3]{wright1997primal}, \cite[Chapter 6]{ben2001lectures}).
Thus, the closed-form control~\Cref{eq:b_def,eq:h_def,eq:u_sol,eq:omega,eq:Omega,eq:mu_sol,eq:lambda_sol,eq:den} has a computational-complexity advantage over the online solution of the quadratic program.

Note that if $\SH$ is compact, then \ref{thm:lip.c} of \Cref{thm:lip} implies that $u$, $\mu$, and $\lambda$ are bounded on $\SH$.

The next theorem is the main result on safety using this control.

\begin{theorem}\label{thm:smocbf}\rm
Consider~\eqref{eq:dynamics}, where \ref{cond:hocbf.a} is satisfied, and consider $u$ given by~\Cref{eq:b_def,eq:h_def,eq:u_sol,eq:omega,eq:Omega,eq:mu_sol,eq:lambda_sol,eq:den}. 
Assume that $h^\prime$ is locally Lipschitz on $\SH$, and for all $x \in \SB$, $L_gh(x) \ne 0$.
Let $x_0\in \SSS$.
Then, for all $t \in I(x_0)$, $x(t) \in \SSS \subseteq \SSS_\rms$.
\end{theorem}

\begin{proof}

Since $h^\prime$ is locally Lipschitz on $\SSS$, it follows from~\ref{thm:lip.d} of~\Cref{thm:lip} that $u$ is locally Lipschitz on $\SH$. 
Next, let $a \in \bd \SH$, which implies that $h(a) =0$.
Thus,~\Cref{eq:u_sol,eq:omega,eq:Omega,eq:mu_sol,eq:lambda_sol,eq:den} imply that 
\begin{align*}
L_fh(a) + L_gh(a) u(a) &= L_fh(a) - L_gh(a) Q(a)^{-1} c(a)\nn\\
&\qquad + L_gh(a) Q(a)^{-1} L_gh(a)^\rmT \lambda(a)\\
&= \begin{cases}
    0, & a\in\Omega\\
     \omega(a), & a\notin\Omega    
    \end{cases}\\
    & \ge 0.
\end{align*}
Hence, for all $x \in \bd \SH$, $L_fh(x) + L_gh(x) u(x) \ge 0$.
Since, in addition, $x_0 \in \SSS \subseteq \SH$, it follows from~\cite[Theorem 4.7]{blanchini2008set} that for all $t \in I(x_0)$, $x(t) \in \SH$, which implies that for all $t \in I(x_0)$, $h(x(t)) \ge 0$.
Thus, \eqref{eq:defS_softmin}, Proposition~\ref{prop:S}, and \eqref{eq:C_j_i_def} imply that for all $t \in I(x_0)$, $b_{j, d_j-1}(x(t)) \ge 0$.
Since, in addition, $x_0 \in \SSS \subseteq \SC_j$, it follows from Lemma~\ref{lemma:fwd_inv} that for all $t \in I(x_0)$, $x(t) \in \SC_j$.
Thus, for all $t \in I(x_0)$, $x(t) \in \SC$, which implies that for all $t \in I(x_0)$, $x(t) \in \SSS = \SH \cap \SC$.
\end{proof}

The control \Cref{eq:b_def,eq:h_def,eq:u_sol,eq:omega,eq:Omega,eq:mu_sol,eq:lambda_sol,eq:den}
relies on the Lie derivatives $L_fh$ and $L_gh$, which can be expressed as 
\begin{align*}
L_fh(x) &= \sum_{j=1}^\ell \beta_j(x) L_f b_{j, d_j-1}(x), \nn\\
& = \sum_{j=1}^\ell \beta_j(x) \Big [ L_f^{d_j} h_j(x) \nn\\
&\qquad + \sum_{k=0}^{d_j-2} L_f^{k+1} \alpha_{j, d_j-2-k}\left (b_{j,d_j-2-k}(x) \right ) \Big ],\\
L_gh(x) &= \sum_{j=1}^\ell \beta_j(x) L_g b_{j, d_j-1}(x) \nn\\
&= \sum_{j=1}^\ell \beta_j(x) L_g L_f^{d_j-1} h_j(x),
\end{align*}
where 
\begin{equation*}
\beta_j(x) \triangleq \exp \rho (h(x) - b_{j, d_j-1}(x)).    
\end{equation*}
It follows from \Cref{eq:softmin,eq:h_def} that $\sum_{j=1}^\ell \beta_j(x) = 1$, which implies that for each $x \in \BBR^n$, $L_fh(x)$ and $L_gh(x)$ are convex combinations of $L_fb_{1, d_1-1}(x),\ldots,L_fb_{\ell, d_\ell-1}(x)$ and $L_gb_{1, d_1-1}(x),\ldots,L_gb_{\ell, d_\ell-1}(x)$, respectively. 
Thus, if this convex combination of $L_gb_{1, d_1-1}(x),\ldots,L_gb_{\ell, d_\ell-1}(x)$ is not equal to the zero vector, then $L_gh(x) \neq 0$.
The next result follows immediately from this observation and provides a sufficient condition such that $L_gh$ is nonzero on $\SB$.

\begin{proposition}\label{prop:Lgh}
\rm{
Assume~\ref{cond:hocbf.a} is satisfied, and assume for all $x\in \SB $, $0 \not \in \mbox{conv} \{L_g L_f^{d_1-1} h_1(x),\ldots,L_g L_f^{d_\ell-1} h_\ell(x)\}$.
Then, for all $x\in\SB$, $L_gh(x) \ne 0$.
}
\end{proposition}

For certain applications, \Cref{prop:Lgh} can be used to verify that $L_gh$ is nonzero on $\SB$. 
For example, if the corresponding elements of $L_g L_f^{d_1-1} h_1(x),\ldots,L_g L_f^{d_\ell-1} h_\ell(x)$ have the same sign, then \Cref{prop:Lgh} implies that  $L_gh$ is nonzero on $\SB$. 
Although \Cref{prop:Lgh} provides a sufficient condition such that $L_gh$ is nonzero on $\SB$, we note that this condition is not necessary.

Next, we present an example to demonstrate the optimal soft-minimum R-CBF control~\Cref{eq:b_def,eq:h_def,eq:u_sol,eq:omega,eq:Omega,eq:mu_sol,eq:lambda_sol,eq:den}.

\begin{example}\label{example:nonholonomic_no_input_const}\rm
Consider the nonholonomic ground robot modeled by~\eqref{eq:dynamics}, where
\begin{equation*}
    f(x) = \begin{bmatrix}
    v \cos{\theta} \\
    v \sin{\theta} \\
    0 \\
    0
    \end{bmatrix}, 
    \,
    g(x) = \begin{bmatrix}
    0 & 0\\
    0 & 0\\
    1 & 0 \\
    0 & 1
    \end{bmatrix}, 
    \,
    x = \begin{bmatrix}
    q_\rmx\\
    q_\rmy\\
    v\\
    \theta
    \end{bmatrix}, 
    \,
    u = \begin{bmatrix}
    u_1\\
    u_2
    \end{bmatrix}, 
\end{equation*}
and $q \triangleq [ \, q_\rmx \quad q_\rmy \, ]^\rmT$ is the robot's position in an orthogonal coordinate system, $v$ is the speed, and $\theta$ is the direction of the velocity vector (i.e., the angle from $[ \, 1 \quad 0 \, ]^\rmT$ to $[ \, \dot q_\rmx \quad \dot q_\rmy \, ]^\rmT$).

Consider the map shown in~\Cref{fig:unicycle_hocbf_map}, which has 6 obstacles and a wall. 
For $j\in \{1,\ldots,6\}$, the area outside the $j$th obstacle is modeled as the zero-superlevel set of 
\begin{equation}\label{eq:h_obs}
h_j(x) = \left \| \matls a_{\rmx, j} (q_\rmx - b_{\rmx,j})  \\ a_{\rmy, j} (q_\rmy - b_{\rmy, j}) \matrs \right \|_p - c_j,    
\end{equation}
where $b_{\rmx,j}, b_{\rmy,j}, a_{\rmx,j}, a_{\rmy,j}, c_j, p> 0$ specify the location and dimensions of the $j$th obstacle.
Similarly, the area inside the wall is modeled as the zero-superlevel set of
\begin{equation}\label{eq:h_wall}
h_7(x) = c_7 - \left \| \matls a_{\rmx, 7} q_\rmx  \\ a_{\rmy, 7} q_\rmy \matrs \right \|_p,   
\end{equation}
where $a_{\rmx, 7}, a_{\rmy, 7}, c_7 , p> 0$ specify the dimension of the space inside the wall. 
The bounds on speed $v$ are modeled as the zero-superlevel sets of
\begin{equation}
    h_8(x) = 9 - v,\qquad 
    h_9(x) = v + 1.\label{eq:h_speed.b}
\end{equation}
The safe set $\SSS_\rms$ is given by \eqref{eq:S_s} where $\ell = 9$.
The projection of $\SSS_\rms$ onto the $q_\rmx$--$q_\rmy$ plane is shown in~\Cref{fig:unicycle_hocbf_map}. 
Note that for all $x\in\SSS_\rms$, the speed satisfies $v\in[-1,9]$.
We also note that~\ref{cond:hocbf.a} is satisfied with $d_1 = d_2 = \ldots = d_7 = 2$ and $d_8 = d_9 = 1$.

\begin{figure}[t!]
\center{\includegraphics[width=0.49\textwidth,clip=true,trim= 0.in 0.in 0in 0in] {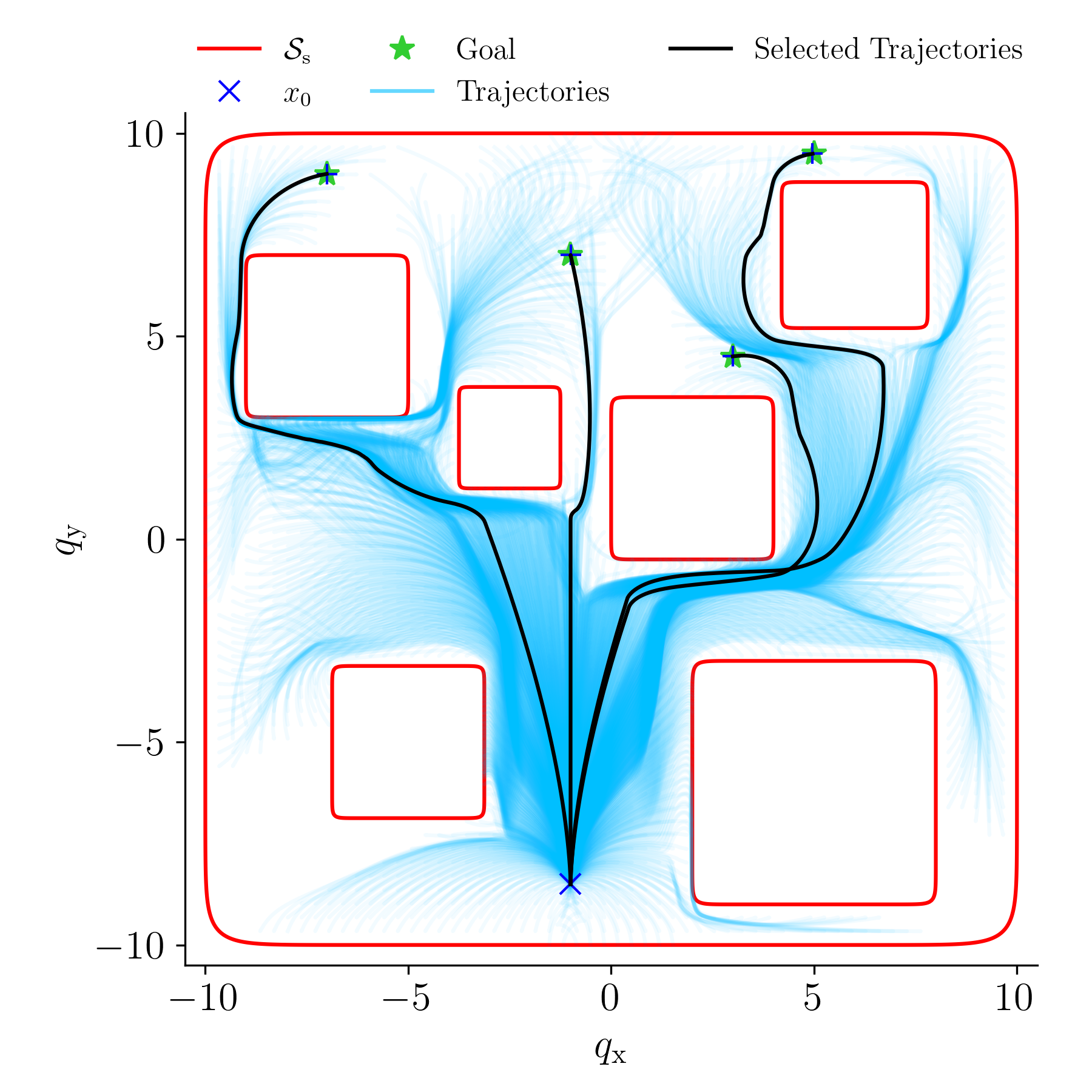}}
\caption{Safe set $\SSS_\rms$, and 2,500 closed-loop trajectories using the control~\Cref{eq:b_def,eq:h_def,eq:u_sol,eq:omega,eq:Omega,eq:mu_sol,eq:lambda_sol,eq:den}.}\label{fig:unicycle_hocbf_map}
\end{figure}

Let $q_{\rmd} = [\, q_{\rmd,\rmx} \quad q_{\rmd,\rmy}\,]^\rmT \in \BBR^2$ be the goal location, that is, the desired location for $q$.
Then, consider the desired control
\begin{equation}\label{eq:u_d_ex}
    u_\rmd(x)\triangleq \begin{bmatrix}
    u_{\rmd_1}(x)\\u_{\rmd_2}(x)
\end{bmatrix},
\end{equation} 
where $u_{\rmd_1},u_{\rmd_2},\psi:\BBR^4 \to \BBR$ are 
\begin{align}
    u_{\rmd_1}(x) &\triangleq - (k_1 + k_3) v + (1+k_1 k_3) \| q-q_\rmd\|_2 \cos \psi(x)\nn \\
    &\qquad + k_1 \left(k_2 \| q-q_\rmd\|_2 + v\right) \sin^2 \psi(x), \label{eq:u_d_1_ex}\\
    u_{\rmd_2}(x) &\triangleq \left(k_2 + \frac{v}{\| q-q_\rmd\|_2}\right) \sin \psi(x), \label{eq:u_d_2_ex}\\
    \psi(x) &\triangleq \mbox{atan2}(q_\rmy-q_{\rmd,\rmy},q_\rmx-q_{\rmd,\rmx})-\theta + \pi, \label{eq:psi}
  \end{align} 
and $k_1 = 0.2, k_2= 1$, and $k_3 = 2$. 
Note that the desired control is designed using a process similar to \cite[pp. 30--31]{de2002control}, and it drives $[ \, q_\rmx \quad q_\rmy \, ]^\rmT$ to $q_\rmd$ but does not account for safety.

We consider the cost~\eqref{eq:quad_cost}, where $Q(x)= I_2$ and $c(x) = -u_\rmd(x)$. 
Thus, the minimizer of~\eqref{eq:quad_cost} is equal to the minimizer $u_\rmd(x)$ of the minimum-intervention cost $\|u - u_\rmd(x)\|_2^2$.

We implement the control~\Cref{eq:b_def,eq:h_def,eq:u_sol,eq:omega,eq:Omega,eq:mu_sol,eq:lambda_sol,eq:den} with $\rho = 10$, $\gamma = 10^{24}$, 
$\alpha_{1,0}(h) = \ldots= \alpha_{7,0}(h) = 10h$, and $\alpha(h) = 0.5h$.
In this example, we use large $\gamma$ (i.e., $\gamma=10^{24}$) to enforce the constraint 
\eqref{eq:qp_softmin.b} without relying on the slack variable (i.e., $\mu \approx 0$). 
Smaller $\gamma$ can also be used to obtain more aggressive behavior, where the closed-loop response gets closer to the obstacles or speed limits. 
However, the choice of $\alpha_{j,0}$ in this example is relatively aggressive (i.e., $\alpha_{j,0}^\prime = 10$).
Hence, if smaller $\gamma$ is used, then $\alpha_{j,0}^\prime$ should be decreased to limit the magnitude of the time derivative of the control signals.
For sample-data implementation, the control is updated at $40~\rm{Hz}$.

Figure~\ref{fig:unicycle_hocbf_map} shows 2,500 closed-loop trajectories for $x_0=[\,-1\quad -8.5\quad 0\quad {\pi}/{2}\,]^\rmT$ and randomly distributed goal locations. 
For each case, all safety constraints are satisfied along the closed-loop trajectory. 
Certain goal locations in $\SSS_\rms$ cannot be reached, which is expected because the desired control \Cref{eq:u_d_ex,eq:u_d_1_ex,eq:u_d_2_ex,eq:psi} does not account for any safety constraints. 
In other words, \Cref{eq:u_d_ex,eq:u_d_1_ex,eq:u_d_2_ex,eq:psi} is not designed using knowledge of the obstacles or speed limits. 
For this example, increasing the gains $k_1,k_2,k_3$ in the desired control \Cref{eq:u_d_ex,eq:u_d_1_ex,eq:u_d_2_ex,eq:psi} enlarges the set of reachable goal locations.
Alternatively, designing a desired control that uses information about the obstacles can also enlarge the set of reachable goal locations.

Figure~\ref{fig:unicycle_hocbf_map} also highlights the closed-loop trajectories for 4 selected goal locations: $q_{\rmd}=[\,3\quad 4.5\,]^\rmT$, $q_{\rmd}=[\,-7\quad 9\,]^\rmT$, $q_{\rmd}=[\,5\quad 9.5\,]^\rmT$, and $q_{\rmd}=[\,-1\quad 7\,]^\rmT$. 
Figures~\ref{fig:unicycle_hocbf_states} and~\ref{fig:unicycle_hocbf_h} show the trajectories of the relevant signals for the case where $q_{\rmd}=[\,3 \quad 4.5\,]^\rmT$.
Figure~\ref{fig:unicycle_hocbf_h} shows that $h$, $\min b_{j,i}$, and $\min h_{j}$ are positive for all time, which implies that the trajectory remains in $\SSS \subseteq \SSS_\rms$.
The control $u$ deviates significantly from $u_\rmd$ at $t = 1.4$ s.
This occurs, in large part, because the desired control \Cref{eq:u_d_ex,eq:u_d_1_ex,eq:u_d_2_ex,eq:psi} does not account for the obstacles or speed limits. 
This effect can be mitigated by using a less aggressive choice for $\alpha_{j,0}$.

\begin{figure}[t!]
\center{\includegraphics[width=0.49\textwidth,clip=true,trim= 0in 0in 0in 0.in] {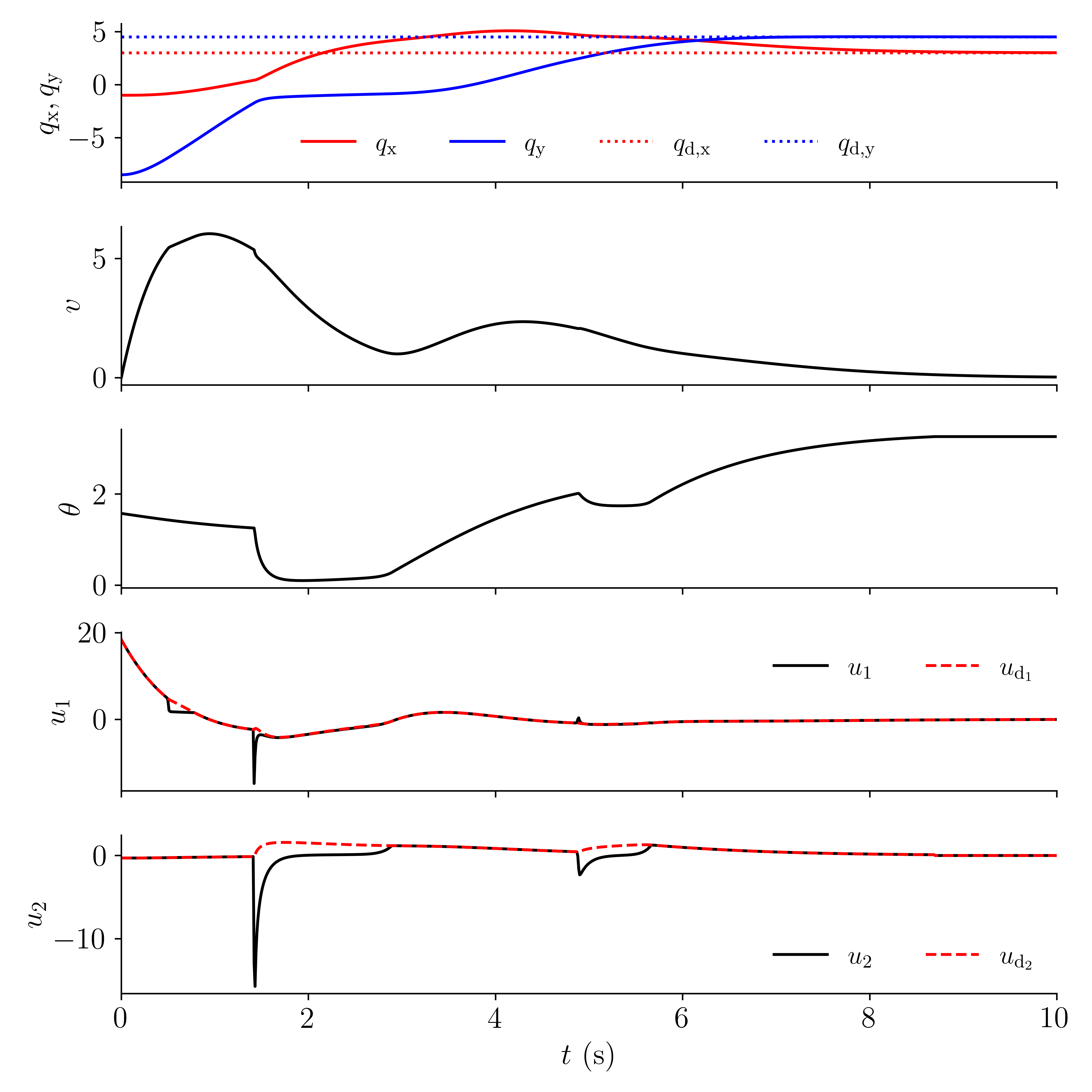}}
\caption{$q_\rmx$, $q_\rmy$, $v$, $\theta$, $u$ and $u_\rmd$ for $q_{\rmd}=[\,3 \quad 4.5\,]^\rmT$.}\label{fig:unicycle_hocbf_states}
\end{figure}

\begin{figure}[t!]
\center{\includegraphics[width=0.49\textwidth,clip=true,trim= 0.in 0in 0in 0in] {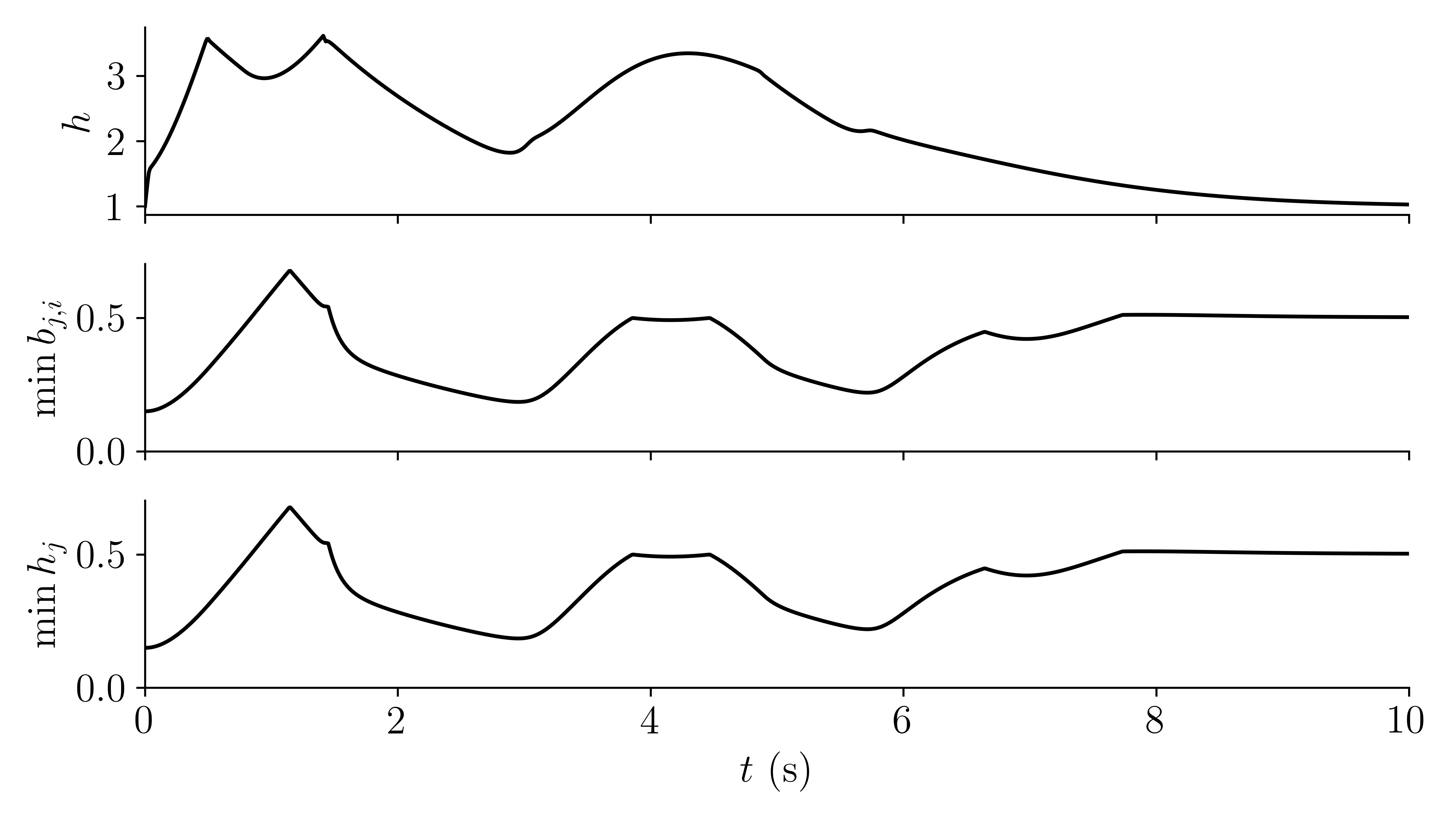}}
\caption{$h$, $\min b_{j,i}$, and $\min h_j$ for $q_{\rmd}=[\,3 \quad 4.5\,]^\rmT$.} \label{fig:unicycle_hocbf_h}
\end{figure}

For comparison, we present simulation results with an alternative control approach that uses multiple higher-order CBFs (HOCBFs).
Specifically, the control is generated from a quadratic program that has multiple constraints---one for each HOCBF (e.g., see~\cite{xiao2021high}).
In this example, $h_1, h_2, \ldots, h_9$ are each used to form individual CBF-based constraints for a quadratic program that uses the minimum-intervention cost.
All parameters are the same as in the above soft-minimum R-CBF implementation. 
Out of the 2,500 random goal locations, the multiple-HOCBF approach resulted in an infeasible quadratic program along 103 of the solution trajectories.
Some of these feasibility violations are minor and did not result in safety violations; however, this illustrates that feasibility is not generally guaranteed using multiple HOCBF-based constraints.
\Cref{ex:compare} in the next section demonstrates that this issue with multiple-HOCBF methods is exacerbated when there are also input constraints.

\Cref{fig:cf_qp_comparison} shows 4 solution trajectories using the soft-minimum R-CBF method and using the multiple-HOCBF method. 
Two of the trajectories shown are selected from the 103 where the multiple-HOCBF method is infeasible.
For these 2 cases, the trajectory ends at the time instant when the optimization is infeasible.

\begin{figure}[t!]
\center{\includegraphics[width=0.49\textwidth,clip=true,trim= 0.in 0.in 0in 0in] {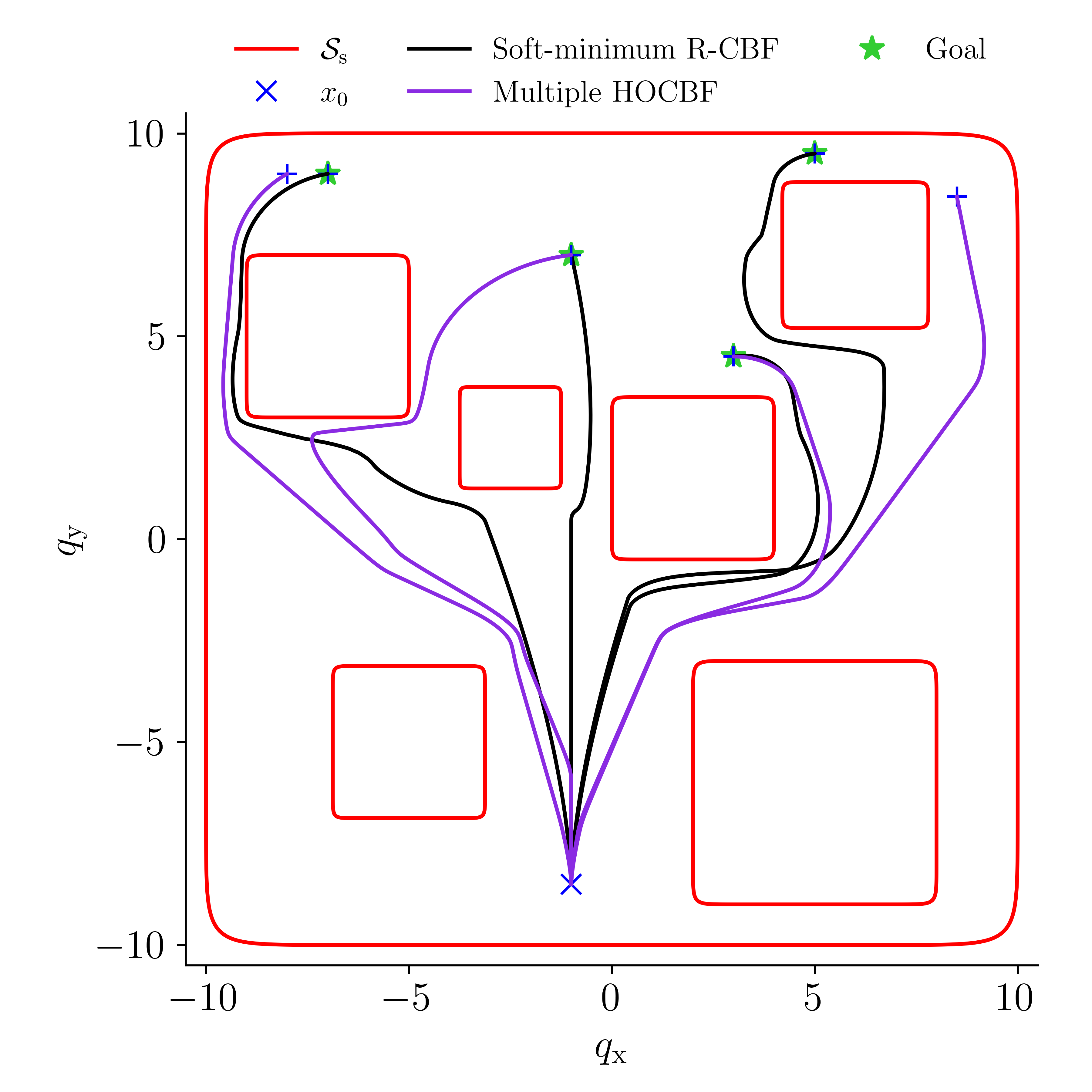}}
\caption{Closed-loop trajectories for 4 goal locations using the soft-minimum R-CBF method and the multiple-HOCBF method.}
\label{fig:cf_qp_comparison}
\end{figure}

Simulations are implemented in Python on a laptop computer with an Intel Core i9-9880H CPU\footnote{\href{https://github.com/pedramrabiee/hocbf\_composition}{https://github.com/pedramrabiee/hocbf\_composition}}. For the multiple-HOCBF approach, the quadratic program is solved using the OptNet package~\cite{amos2017optnet}. 
The average computation time for the soft-minimum R-CBF method is 2~ms per iteration, while it is 30~ms for the multiple-HOCBF approach, which illustrates the computational benefit of the soft-minimum R-CBF method.
In fact, the computational time for the soft-minimum R-CBF method is negligible if $L_fh$ and $Lgh$ are calculated analytically offline rather than in real time.
\exampletriangle
\end{example}

\section{Safety with Input Constraints}
\label{section:input constraints}

This section presents a closed-form optimal control that not only guarantees safety but also respects specified input constraints (e.g., actuator limits).
We reconsider the system \eqref{eq:dynamics}, safe set \eqref{eq:S_s}, and cost \eqref{eq:quad_cost}.
Next, let $\phi_1,\ldots,\phi_{\ell_{u}}: \BBR^m \to \BBR$ be continuously differentiable, and define the set of admissible controls
\begin{equation}\label{eq:U_def}
    \SU \triangleq \{ u\in \BBR^m: \phi_1(u) \ge 0, \ldots,  \phi_{\ell_{ u}}( u) \ge 0\} \subseteq \BBR^m.
\end{equation}
%
We assume that for all $u \in \SU$ and all $\kappa \in \{1, 2, \ldots, \ell_{u}\}$, $\phi^\prime_\kappa(u) \neq 0$.
Unless otherwise stated, all statements in this section that involve the subscript $\kappa$ are for all $\kappa \in \{ 1,2,\ldots,\ell_u\}$.

The control objective is to design a full-state feedback control such that for all $t \in I(x_0)$, $ J(x(t), u(t))$ is minimized subject to the safety constraint that $x(t) \in \SSS_\rms$ and the input constraint that $u(t) \in \SU$.

Before presenting an approach to this general problem, the next example considers a special case where there is one safety constraint that is relative degree one and one input constraint.
The example illustrates the 3 key elements used to address safety with input constraints: (1) control dynamics to transform the input constraint into a controller-state constraint; (2) soft-minimum R-CBF to compose the safety constraint and controller-state constraint; and (3) a desired surrogate control (i.e., desired input to the control dynamics).

\begin{example}
\label{ex:input constraints simplified}
\rm

We consider a scenario, where there is one safety constraint (i.e., $\ell = 1$) that is relative degree one (i.e., $d_1 = 1$) and one input constraint (i.e., $\ell_u =1$). 
Thus, the safety constraint is described by the zero-superlevel set of $h_1$ and the input constraint is described by the zero-superlevel set of $\phi_1$.

First, consider a control $u$ generated by the first-order linear time-invariant (LTI) control dynamics
\begin{equation}
    \dot u(t) = -a u(t) + a \hat u(x(t),u(t)),\label{eq:dynamics_control.a.low_pass}
\end{equation}
where $a > 0$, $u(0) \in \BBR^m$ is the initial condition, and $\hat u: \BBR^n \times \BBR^m \to \BBR^m$ is the \textit{surrogate control} (i.e., input to the control dynamics). 
The cascade of~\cref{eq:dynamics,eq:dynamics_control.a.low_pass} is
\begin{equation}\label{eq:dynamics_aug.1.low_pass}
    \dot{\hat x} = \hat f(\hat x) + \hat g(\hat x) \hat u,
\end{equation}
where
\begin{gather}\label{eq:dynamics_aug.2.low_pass}
\hat x \triangleq  \begin{bmatrix}
            x \\ u
        \end{bmatrix}, \quad 
       \hat f( \hat x) \triangleq \begin{bmatrix}
            f(x) + g(x) u\\
            -a u
        \end{bmatrix},\quad
        \hat g(\hat x) \triangleq \begin{bmatrix}
            0\\
            a
        \end{bmatrix}. 
\end{gather}
Define 
\begin{equation*}
\hat \SSS_\rms \triangleq \SSS_\rms \times \SU,
\end{equation*}
which is the set of system-and-controller states $\hat x$ such that the safety constraint is satisfied (i.e., $h_1(x) \ge 0$) and the input constraint is satisfied (i.e., $\phi_1(u) \ge 0$).
Define 
\begin{equation*}
    \hat h_1(\hat x) \triangleq h_1(x), \qquad \hat h_2(\hat x) \triangleq \phi_1(u),
\end{equation*}
and for $j\in \{1,2\}$, the zero-superlevel sets are $$\hat \SC_{j,0} \triangleq \{ \hat x \in \BBR^{n+m} \colon \hat h_j(\hat x)\geq 0 \}.$$
Thus, $\hat S_\rms = \hat \SC \triangleq \hat \SC_{1,0} \cap \hat \SC_{2,0}$, which demonstrates that the control dynamics \eqref{eq:dynamics_control.a.low_pass} transforms the input constraint into a controller-state constraint
However, $\hat h_1$ and $\hat h_2$ have different relative degrees with respect to the cascade~\cref{eq:dynamics_aug.1.low_pass,eq:dynamics_aug.2.low_pass}; specifically, the relative degrees are $\hat d_1 = 2$ and $\hat d_2 = 1$.
Thus, the introducing control dynamics also increases the relative degree of the original safety constraint. 
Nevertheless, the soft-minimum R-CBF construction can be used to construct a single R-CBF from $\hat h_1$ and $\hat h_{2}$.

To use the soft-minimum R-CBF to compose $\hat h_1$ (safety constraint) and $\hat h_2$ (controller-state constraint), we follow the steps in \Cref{section:composite CBF} but where each symbol is replaced by that symbol with a hat (i.e., $\hat \cdot$).
Specifically, we define
\begin{gather*}
\hat b_{1,0}(\hat x) \triangleq \hat h_1(\hat x), \qquad \hat b_{2,0}(\hat x) \triangleq \hat h_2(\hat x),\\
\hat b_{1,1}(\hat x) \triangleq L_{\hat f} \hat h_1(\hat x) + \hat \alpha_{1,0}(\hat h_1(\hat x))
\end{gather*}
where $\hat \alpha_{1,0} \colon \BBR \to \BBR$ is a locally Lipschitz extended class-$\SK$ function. 
Then, the composite soft-minimum R-CBF is 
\begin{equation*}
    \hat h(\hat x) \triangleq \mbox{softmin}_\rho \Big ( \hat b_{1, 1}(\hat x), \hat b_{2, 0}(\hat x)\Big ),
\end{equation*}
and its zero-superlevel set is 
\begin{equation*}
\hat \SH \triangleq \{ \hat x \in \BBR^{n+m} \colon \hat h( \hat x) \ge 0 \}.
\end{equation*}
Next, \Cref{lem:rcbf} implies that $\hat h$ is an R-CBF if for all $x \in \hat \SB \triangleq \{ x \in \bd \hat \SH \colon L_{\hat f} \hat h(\hat x) \le 0 \}$, $L_{\hat g} \hat h(\hat x) \ne 0$. 
Thus, we aim to use $\hat h$ to construct a closed-form optimal surrogate control $\hat u$ such that $\hat x(t) \in \hat \SSS \triangleq \hat \SH \cap \hat \SC$, which implies that the safety and input constraints are satisfied. 
However, we cannot directly apply the process in \Cref{section:control} because the cost $J$ given by~\eqref{eq:quad_cost} is a function of $u$ rather than $\hat u$.

To address this issue, we define the \textit{desired surrogate control}
\begin{align}\label{eq:u_d_hat_def.low_pass}
    \hat u_\rmd(\hat x) &\triangleq u + \frac{1}{a} \Bigg( u_\rmd^\prime(x) \Big( f(x) + g(x) u \Big) + \sigma_0 \Big(u_\rmd(x) - u  \Big) \Bigg),
\end{align}
where $\sigma_0>0$ and $u_\rmd(x) \triangleq - Q(x)^{-1} c(x)$, which is the unconstrained minimizer of~\eqref{eq:quad_cost}.
We note that $\hat u_\rmd$ is a feedback linearizing input for the cascade~\cref{eq:dynamics_aug.1.low_pass,eq:dynamics_aug.2.low_pass}; it makes the error between the control $u$ and the desired control $u_\rmd$ satisfy asymptotically stable LTI dynamics.
More specifically, if $\hat u= \hat u_\rmd$, then substituting \eqref{eq:u_d_hat_def.low_pass} into \eqref{eq:dynamics_control.a.low_pass} yields
\begin{equation*}
   \frac{\rmd}{\rmd t} \left [ u(t) - u_\rmd(x(t)) \right ] + \sigma_0 \left [ u(t) - u_\rmd(x(t)) \right ] = 0.
\end{equation*} 
Hence, the desired surrogate control \eqref{eq:u_d_hat_def.low_pass} yields $u$ that converges exponentially to $u_\rmd$, which is the minimizer of $J$. 
Next, we consider the \textit{surrogate cost function}
\begin{equation*}
    \hat J(\hat x,\hat u) \triangleq \frac{1}{2} \| \hat u - \hat u_\rmd(\hat x) \|_2^2,
\end{equation*}
and we can obtain an optimal surrogate control $\hat u$ that satisfies safety and input constraints by solving the quadratic program \eqref{eq:qp_softmin}, where $h$, $L_f h$, $L_g h$, and $J$ are replaced by $\hat h$, $L_{\hat f} \hat h$, $L_{\hat g} \hat h$, and $\hat J$. 
This results in a closed-form for $\hat u$ that is given by~\Cref{eq:u_sol,eq:omega,eq:Omega,eq:mu_sol,eq:lambda_sol,eq:den}, where $Q=I_m$, $c=-\hat u_\rmd$, and $h$, $L_f h$, and $L_g h$ are replaced by $\hat h$, $L_{\hat f} \hat h$, and $L_{\hat g} \hat h$. 
\exampletriangle
\end{example}

The remainder of this section formalizes and generalizes the approach demonstrated in \Cref{ex:input constraints simplified}.
We generalize to an arbitrary number of safety constraints with potentially different relative degrees, an arbitrary number of input constraints, and a general class of nonlinear control dynamics.
Moreover, we analyze the properties of this approach.

\subsection{Control Dynamics to Transform Input Constraints into Controller-State Constraints}

To address safety with input constraints, consider a control $u$ that satisfies
\begin{align}
    \dot x_\rmc(t) &= f_\rmc(x_\rmc(t)) + g_\rmc(x_\rmc(t)) \hat u(x(t),x_\rmc(t)),\label{eq:dynamics_control.a}\\
    u(t) &= h_\rmc(x_\rmc(t)),\label{eq:dynamics_control.b}
\end{align}
where $x_\rmc(t) \in \BBR^{n_\rmc}$ is the controller state; $x_\rmc(0) = x_{\rmc0} \in \BBR^{n_\rmc}$ is the initial condition; $f_\rmc:\BBR^{n_\rmc}\to\BBR^{n_\rmc}$, $g_\rmc:\BBR^{n_\rmc}\to\BBR^{{n_\rmc}\times m}$, and $h_\rmc:\BBR^{n_\rmc}\to\BBR^m$ are locally Lipschitz on $\BBR^{n_\rmc}$; and $\hat u: \BBR^{n+n_\rmc} \to \BBR^m$ is the \textit{surrogate control} (i.e., input to the control dynamics), which is given by the closed-form solution to a quadratic program presented later in this section.

Define 
\begin{equation}\label{eq:S_c_def}
\SSS_{\rmc} \triangleq \{x_\rmc \in \BBR^{n_\rmc}:h_\rmc(x_\rmc) \in \SU\},
\end{equation}
which is the set of controller states such that the control is in $\SU$. 
Thus, the control dynamics~\cref{eq:dynamics_control.a,eq:dynamics_control.b} transforms the input constraint (i.e., $u(t) \in \SU$) into a constraint on the controller state (i.e., $x_\rmc(t) \in \SSS_\rmc$).

Next, the cascade of~\cref{eq:dynamics,eq:dynamics_control.a,eq:dynamics_control.b} is given by 
\begin{equation}\label{eq:dynamics_aug.1}
    \dot{\hat x} = \hat f(\hat x) + \hat g(\hat x) \hat u,
\end{equation}
where
\begin{gather}\label{eq:dynamics_aug.2}
    \qquad \hat x \triangleq  \begin{bmatrix}
            x \\ x_\rmc
        \end{bmatrix}, \qquad 
       \hat f( \hat x) \triangleq \begin{bmatrix}
            f(x) + g(x)h_\rmc(x_\rmc)\\
            f_\rmc(x_\rmc)
        \end{bmatrix},\\
        \hat g(\hat x) \triangleq \begin{bmatrix}
            0\\
            g_\rmc(x_\rmc)
        \end{bmatrix}, \qquad \hat x_0 \triangleq  \begin{bmatrix}
            x_0 \\ x_{\rmc 0}
        \end{bmatrix}, \label{eq:dynamics_aug.3}
\end{gather}
and $\hat n\triangleq n + n_\rmc$.
Define 
\begin{equation}\label{eq:hat Ss}
\hat \SSS_\rms \triangleq \SSS_\rms \times \SSS_{\rmc},
\end{equation}
which is the set of cascade states $\hat x$ such that the safety constraint (i.e., $x(t) \in \SSS_\rms$) and the input constraint (i.e., $u(t) \in \SU$) are satisfied. 
The next result summarizes this property. 
The proof is in Appendix~\ref{app:input_constraint_propositions}.

\begin{proposition}\label{prop:cascade_constraints}\rm
Assume that for all $t \in I(\hat x_0)$, $\hat x(t) \in \hat \SSS_\rms$.
Then, for all $t \in I(\hat x_0)$, $x(t) \in \SSS_\rms$ and $u(t) \in \SU$.
\end{proposition}

The functions $f_\rmc$, $g_\rmc$ and $h_\rmc$ are selected such that the following conditions hold:
\begin{enumC}

    \item \label{cond:hocbf_x_c.a}
    There exists a positive integer $d_\rmc$ such that for all $x_\rmc \in \BBR^{n_{\rmc}}$ and all $i \in \{0,1,\ldots,d_\rmc-2\}$, $L_{g_\rmc} L_{f_\rmc}^i h_\rmc(x_\rmc) = 0$; and for all $x_\rmc \in \BBR^{n_{\rmc}}$, $L_{g_\rmc} L_{f_\rmc}^{d_\rmc-1}h_\rmc(x_\rmc)$ is nonsingular.

    \item \label{cond:hocbf_x_c.d}
    There exists a positive integer $\zeta$ such that for  all $x_\rmc \in \BBR^{n_{\rmc}}$ and all $i \in \{0,1,\ldots,\zeta-2\}$, $L_{g_\rmc} L_{f_\rmc}^i \phi_\kappa(h_\rmc(x_\rmc)) = 0$; and for  all $x_\rmc \in \SSS_\rmc$, $L_{g_\rmc} L_{f_\rmc}^{\zeta-1}\phi_\kappa(h_\rmc(x_\rmc)) \neq 0$.

\end{enumC}

The following example provides one construction for $f_\rmc$, $g_\rmc$, and $h_\rmc$ such that~\ref{cond:hocbf_x_c.a} and \ref{cond:hocbf_x_c.d} are satisfied. 
In particular, the example demonstrates that~\ref{cond:hocbf_x_c.a} and \ref{cond:hocbf_x_c.d} are satisfied by any LTI construction of~\cref{eq:dynamics_control.a,eq:dynamics_control.b} where the first Markov parameter is nonsingular.

\begin{example}\label{ex:control dynamics}
\rm{
Let 
\begin{equation*}
f_\rmc(x_\rmc) = A_\rmc x_\rmc, \quad g_\rmc(x_\rmc) = B_\rmc, \quad h_\rmc(x_\rmc) = C_\rmc x_\rmc, 
\end{equation*}
where $A_\rmc \in \BBR^{n_\rmc \times n_\rmc}$, $B_\rmc \in \BBR^{n_\rmc \times m}$, $C_\rmc \in \BBR^{m\times n_\rmc}$, and $C_\rmc B_\rmc$ is nonsingular. 
Note that $L_{g_\rmc}h_\rmc(x_\rmc) = C_\rmc B_\rmc$, which implies that \ref{cond:hocbf_x_c.a} is satisfied with $d_\rmc = 1$. 
Next, note that $L_{g_\rmc}\phi_\kappa(h_\rmc(x_\rmc)) = \phi_\kappa^\prime(h_\rmc(x_\rmc)) C_\rmc B_\rmc$.
Since $C_\rmc B_\rmc$ is nonsingular and for all $x_\rmc \in \SSS_\rmc$, 
$\phi_\kappa^\prime(h_\rmc(x_\rmc))\neq 0$, it follows that \ref{cond:hocbf_x_c.d} is satisfied with $\zeta= 1$. 
Thus, any LTI controller with $C_\rmc B_\rmc$ nonsingular satisfies \ref{cond:hocbf_x_c.a} and \ref{cond:hocbf_x_c.d}.
In fact, any LTI controller where the first nonzero Markov parameter is nonsingular satisfies \ref{cond:hocbf_x_c.a} and \ref{cond:hocbf_x_c.d}.

The control dynamics \eqref{eq:dynamics_control.a.low_pass} used in \Cref{ex:input constraints simplified} is a special case where $A_\rmc = -a I_{n_\rmc}$, $B_\rmc = a I_{n_\rmc}$, $C_\rmc = I_{n_\rmc}$, which results in low-pass control dynamics.}
\exampletriangle
\end{example}

Next, we examine the relative degree of each safety and input constraint function with respect to the cascade~\cref{eq:dynamics_aug.1,eq:dynamics_aug.2,eq:dynamics_aug.3}.
Unless otherwise stated, all statements in this section that involve the subscript $\hat \jmath$ are for all $\hat \jmath \in \{1, 2, \ldots, \hat \ell\}$, where $\hat \ell \triangleq \ell + \ell_{ u}$. 
Let $\hat h_{\hat \jmath} \colon \BBR^{\hat n} \to \BBR$ be defined by 
\begin{equation}\label{eq:hat h}
   \hat  h_{\hat \jmath}( \hat x) \triangleq \begin{cases}
         h_{\hat \jmath}( x), & \mbox{if } {\hat \jmath}\in\{1,2,\ldots, \ell \},\\
        \phi_{{\hat \jmath} - \ell}(h_\rmc(x_\rmc)), & \mbox{if } {\hat \jmath}\in\{\ell+1,\ell+2,\ldots, \ell+\ell_u\},
    \end{cases}
\end{equation}
and define
\begin{equation}\label{eq:hat d}
    \hat d_{\hat \jmath} \triangleq \begin{cases}
        d_{\hat \jmath} + d_\rmc, & \mbox{if } {\hat \jmath}\in\{1,2,\ldots, \ell\},\\
        \zeta, & \mbox{if } {\hat \jmath}\in\{\ell+1,\ell+2,\ldots, \ell+\ell_u\}.
    \end{cases}
\end{equation}

The following result demonstrates that
 $\hat h_{\hat \jmath}$ has well-defined relative degree $\hat d_{\hat \jmath}$ with respect to the cascade~\cref{eq:dynamics_aug.1,eq:dynamics_aug.2,eq:dynamics_aug.3} on $\hat \SSS_\rms$; however, relative degrees $\hat d_1,\ldots,\hat d_{\hat \ell}$ are not all equal.
 The proof is in Appendix~\ref{app:input_constraint_propositions}.

\begin{proposition}\label{prop:cascade}\rm
Consider~\eqref{eq:dynamics}, where~\ref{cond:hocbf.a} is satisfied, and consider~\cref{eq:dynamics_control.a,eq:dynamics_control.b}, where~\ref{cond:hocbf_x_c.a} and \ref{cond:hocbf_x_c.d} are satisfied.
Then, for all $\hat x \in \BBR^{\hat n}$, all $\hat \jmath \in \{ 1,2,\ldots,\hat \ell \}$, and all $i \in \{ 0,1,\ldots, \hat d_{\hat \jmath}-2 \}$, $L_{\hat g} L_{\hat f}^i \hat h_{\hat \jmath}(\hat x) = 0$; and for all $\hat x \in \hat \SSS_\rms$ and all $\hat \jmath \in \{ 1,2,\ldots,\hat \ell \}$, $L_{\hat g}L_{\hat f}^{\hat d_{\hat \jmath}-1} \hat h_{\hat \jmath}(\hat x) \neq 0$. 
\end{proposition}

\Cref{prop:cascade} illustrates that the control dynamics~\cref{eq:dynamics_control.a,eq:dynamics_control.b} increase the relative degree of each system-state constraint (i.e., safety constraint) by $d_\rmc$.
In general, it is desirable to design~\cref{eq:dynamics_control.a,eq:dynamics_control.b} such that $d_\rmc$ is low, which tends to limit lag. 
\Cref{ex:control dynamics} demonstrates that it is always possible to design~\cref{eq:dynamics_control.a,eq:dynamics_control.b} such that $d_\rmc=1$.

\subsection{Composite Soft-Minimum R-CBF}

\Cref{prop:cascade} implies that the cascade~\cref{eq:dynamics_aug.1,eq:dynamics_aug.2,eq:dynamics_aug.3} satisfy~\ref{cond:hocbf.a}, where $f$, $g$, $x$, $\SSS_\rmc$, $\ell$, $j$, $h_j$, and $d_j$ are replaced by $\hat f$, $\hat g$, $\hat x$, $\hat \SSS_\rmc$, $\hat \ell$, $\hat \jmath$, $\hat h_{\hat \jmath}$, and $\hat d_{\hat \jmath}$.
Thus, the composite soft-minimum R-CBF construction in Section~\ref{section:composite CBF} can be applied to the cascade~\cref{eq:dynamics_aug.1,eq:dynamics_aug.2,eq:dynamics_aug.3} to construct a single R-CBF from $\hat h_1,\ldots,\hat h_{\hat \ell}$, which do not all have the same relative degree. 
Note that $\hat h_1,\ldots,\hat h_{\hat \ell}$ describe the set $\hat \SSS_\rms = \SSS_\rms \times \SSS_\rmc$ that combines the safe set $\SSS_\rms$ and the set $\SSS_\rmc$ of controller states such that the control is in $\SU$. 
Thus, the soft-minimum R-CBF can address safety constraints and input constraints.

We construct this soft-minimum R-CBF by following 
\cref{eq:b_def,eq:C_j_i_def,eq:C_j_def,eq:C_def,eq:h_def,eq:defS_softmin,eq:S,eq:SB} in \Cref{section:composite CBF} but where each symbol is replaced by that symbol with a hat (i.e., $\hat \cdot$).
Specifically, let $\hat b_{{\hat \jmath},0}(\hat x) \triangleq \hat h_{\hat \jmath}(\hat x)$. 
For $i\in\{0, 1\ldots,\hat d_{\hat \jmath}-2\}$, let $\hat \alpha_{{\hat \jmath},i} \colon \BBR \to \BBR$ be a locally Lipschitz extended class-$\SK$ function, and consider $\hat b_{{\hat \jmath}, i+1}:\BBR^{\hat n} \to \BBR$ defined by
\begin{equation}\label{eq:b_def.2}
    \hat b_{{\hat \jmath},i+1}(\hat x) \triangleq L_{\hat f} \hat b_{{\hat \jmath},i}(\hat x) + \hat \alpha_{{\hat \jmath},i}(\hat b_{{\hat \jmath},i}(\hat x)).
\end{equation}
For $i \in \{0\ldots,\hat d_{\hat \jmath}-1\}$, define
\begin{equation*}
    \hat \SC_{{\hat \jmath},i} \triangleq \{\hat x\in \BBR^{\hat n} \colon \hat b_{{\hat \jmath},i}(\hat x) \ge 0\}.
\end{equation*}
Next, define
\begin{gather*}  
    \hat \SC_{\hat \jmath} \triangleq 
        \begin{cases}
        \hat \SC_{{\hat \jmath},0}, & \hat d_{\hat \jmath} = 1,\\
    \bigcap_{i=0}^{\hat d_{\hat \jmath} - 2} \hat \SC_{{\hat \jmath},i}, & \hat d_{\hat \jmath} > 1,
    \end{cases}
\end{gather*}
and 
\begin{equation*}  
    \hat \SC \triangleq \bigcap_{{\hat \jmath}=1}^{\hat \ell} \hat \SC_{\hat \jmath}.
\end{equation*}
Let $\rho > 0$, and consider $\hat h:\BBR^{\hat n} \to \BBR$ defined by
\begin{equation}\label{eq:h_def.2}
    \hat h(\hat x) \triangleq \mbox{softmin}_\rho \Big ( \hat b_{1, \hat d_1-1}(\hat x), \hat b_{2, \hat d_2-1}(\hat x), \ldots,\hat b_{\hat \ell, \hat d_{\hat \ell}-1}(\hat x) \Big ).
\end{equation}
Furthermore, define
\begin{gather*}
\hat \SH \triangleq \{ \hat x \in \BBR^{\hat n} \colon \hat h( \hat x) \ge 0 \},\qquad \hat \SSS\triangleq \hat \SH \cap \hat \SC,\\
\hat \SB \triangleq \{ \hat x \in \bd \hat \SH \colon L_{\hat f} \hat h(\hat x) \le 0 \}. 
\end{gather*}
\Cref{lem:rcbf} implies that if $L_{\hat g}\hat h$ is nonzero on $\hat \SB$, then $\hat h$ is an R-CBF.
The next result is a corollary of \Cref{prop:cont_fwd_S_C}, which is obtained by applying \Cref{prop:cont_fwd_S_C} to the cascade~\cref{eq:dynamics_aug.1,eq:dynamics_aug.2,eq:dynamics_aug.3}.

\begin{corollary} \label{prop:cont_fwd_S_C.2} 
\rm{
Consider~\eqref{eq:dynamics}, where~\ref{cond:hocbf.a} is satisfied, and consider~\cref{eq:dynamics_control.a,eq:dynamics_control.b}, where~\ref{cond:hocbf_x_c.a} and \ref{cond:hocbf_x_c.d} are satisfied.
Assume that $\hat h^\prime$ is locally Lipschitz on $\hat \SH$, and for all $\hat x \in \hat \SB$, $L_{\hat g} \hat h(\hat x) \ne 0$.
Then, $\hat \SSS$ is control forward invariant.
}
\end{corollary}

\Cref{prop:cont_fwd_S_C.2} provides conditions under which $\hat \SSS \subseteq \hat \SSS_\rms$ is control forward invariant.
In this case, $\hat h$ is an R-CBF that can be used to generate a control such that for all $t \in I(\hat x_0)$, $\hat x(t) \in \hat \SSS \subseteq \hat \SSS_\rms$, which implies that $x(t) \in \SSS_\rms$ and $u(t) \in \SU$.

\subsection{Desired Surrogate Control}\label{section:surrogate}

Although $\hat h$ is an R-CBF, we cannot directly apply a quadratic program similar to~\eqref{eq:qp_softmin} because the cost $J$ given by~\eqref{eq:quad_cost} is a function of $u$ rather than the surrogate control $\hat u$.
Consider the \textit{desired control} $u_\rmd:\BBR^{\hat n} \to \BBR^m$ defined by
\begin{align}\label{eq:u_d_def}
u_\rmd(\hat x) \triangleq - Q(x)^{-1} c(x),
\end{align}
which is the minimizer of~\eqref{eq:quad_cost}.
Next, consider the \textit{desired surrogate control} $\hat u_\rmd: \BBR^{\hat n} \to \BBR^m$ defined by
\begin{align}\label{eq:u_d_hat_def}
    \hat u_\rmd(\hat x) \triangleq &\left(L_{g_\rmc} L_{f_\rmc}^{d_\rmc - 1} h_\rmc(x_\rmc)\right)^{-1} \Bigg( L_{\hat f}^{d_\rmc} u_\rmd(\hat x)-L_{f_\rmc}^{d_\rmc} h_\rmc(x_\rmc)\nn\\
    &\qquad+\sum_{i = 0}^{d_\rmc-1} \sigma_i \Big(L_{\hat f}^{i} u_\rmd(\hat x)-L_{f_\rmc}^i h_\rmc(x_\rmc)\Big)\Bigg),
\end{align}
where $\sigma_0, \sigma_1, \ldots, \sigma_{d_\rmc -1}>0$ are selected such that 
\begin{equation*}
\sigma(s) \triangleq s^{d_\rmc}+ \sigma_{d_\rmc -1} s^{d_\rmc -1} + \sigma_{d_\rmc -2} s^{d_\rmc - 2}+ \cdots + \sigma_1 s + \sigma_0
\end{equation*}
has all its roots in the open left-hand complex plane.

The desired surrogate control \eqref{eq:u_d_hat_def} is the generalization of \eqref{eq:u_d_hat_def.low_pass} to any $f_\rmc$, $g_\rmc$, $h_\rmc$, and $d_\rmc$ that satisfy~\ref{cond:hocbf_x_c.a} and \ref{cond:hocbf_x_c.d}.
In particular, \eqref{eq:u_d_hat_def} is a feedback linearizing input for the cascade~\cref{eq:dynamics_aug.1,eq:dynamics_aug.2,eq:dynamics_aug.3}; it makes the error between $u$ and $u_\rmd$ satisfy asymptotically stable LTI dynamics.
The next result formalizes this fact. 
The proof is in Appendix~\ref{app:input_constraint_propositions}.

\begin{proposition}\label{proposition:error}
\rm{
Consider~\eqref{eq:dynamics}, where~\ref{cond:hocbf.a} is satisfied, and consider~\cref{eq:dynamics_control.a,eq:dynamics_control.b}, where~\ref{cond:hocbf_x_c.a} is satisfied and $\hat u(\hat x)= \hat u_\rmd(\hat x)$. 
Then, the following statements hold:
\begin{enumalph}

\item \label{proposition:error.a}
The error $u - u_\rmd(\hat x)$ satisfies
\begin{align} 
& \frac{\rmd^{d_\rmc}}{\rmd t^{d_\rmc}}  \left [ u(t) - u_\rmd(\hat x(t)) \right ] \nn\\ 
&\qquad + \sum_{i=0}^{d_\rmc-1}  \sigma_{i} \frac{\rmd^i}{\rmd t^i} \left [ u(t) - u_\rmd(\hat x(t)) \right ] = 0.\label{eq:sum_ue}
\end{align} 

\item \label{proposition:error.b}
For all $\hat x_0 \in \BBR^{\hat n}$, $\lim_{t \to \infty} \left [ u(t) - u_\rmd(\hat x(t)) \right ] =0$ exponentially.

\item \label{proposition:error.c}
Assume that $Q$ is bounded. 
Then, for all $\hat x_0 \in \BBR^{\hat n}$,
    \begin{equation*}
    \lim_{t \to \infty} \Big [ J(x(t), u(t)) - J(x(t), u_\rmd(\hat x(t)))\Big ] = 0.
    \end{equation*}

\item \label{proposition:error.d}
Let $x_0\in \BBR^{n}$, and assume $x_{\rmc 0}\in \BBR^{n_\rmc}$ is such that for $i \in \{0,1,\ldots, d_\rmc -1\}$, $L_{f_\rmc}^i h_\rmc(x_{\rmc 0}) = L_{\hat f}^i u_\rmd(\hat x_0)$.
Then, $u(t) \equiv u_\rmd(\hat x(t))$.

\end{enumalph}
}
\end{proposition}

\begin{figure*}[ht]
\center{\includegraphics[width=\textwidth,clip=true,trim= 0.0in 0.0in 0.0in 0.0in] {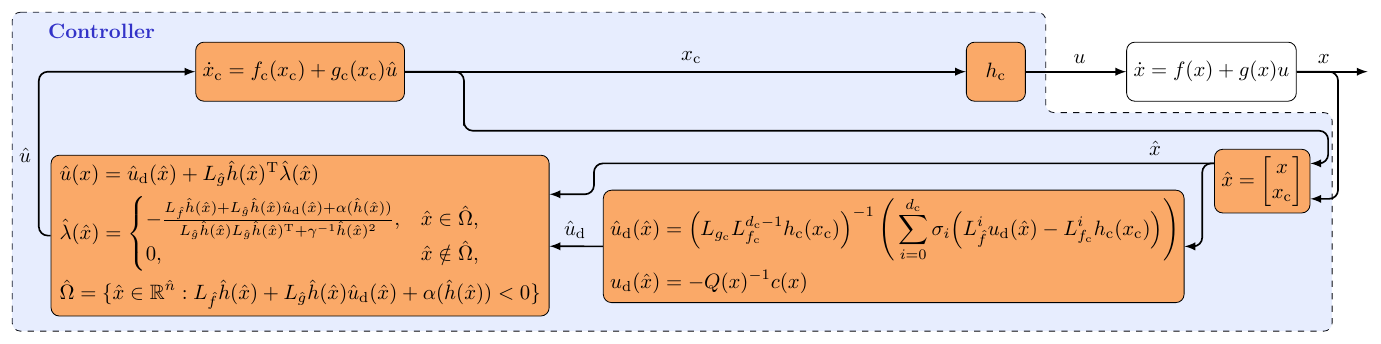}}
\caption{Closed-form optimal and safe control with input constraints. 
Control uses the composite soft-minimum R-CBF $\hat h$ to guarantee safety and input constraints are satisfied. 
Control minimizes cost subject to safety and input constraints.
}\label{fig:block_diag}
\end{figure*}

\Cref{proposition:error} implies that the control dynamics~\cref{eq:dynamics_control.a,eq:dynamics_control.b} with $\hat u=\hat u_\rmd$ yields a control $u$ that converges exponentially to the desired control $u_\rmd$, which is the minimizer of $J$. 
Thus, we introduce the surrogate cost  function
$\hat J \colon \BBR^{\hat n} \times \BBR^m\to \BBR$ defined by
\begin{equation}\label{eq:surrogate cost}
    \hat J(\hat x,\hat u) \triangleq \frac{1}{2} \| \hat u - \hat u_\rmd(\hat x) \|_2^2.
\end{equation}
Note that if $\hat u$ equals the minimizer $\hat u_\rmd$ of $\hat J$, then $u$ converges exponentially to the minimizer $u_\rmd$ of $J$.
In the next subsection, we design a control $\hat u$ with the goal that for all $t \in I(\hat x_0)$, $\hat J(\hat x(t), \hat u( \hat x(t))$ is minimized subject to the constraint that $\hat x(t) \in \hat \SSS \subseteq \hat \SSS_\rms$, which implies that the safety constraints and the input constraints are satisfied.

\subsection{Closed-Form Optimal and Safe Control with Input Constraints}

Let $\gamma > 0$, and let $\alpha:\BBR \to \BBR$ be a locally Lipschitz nondecreasing function such that $\alpha(0) = 0$. 
For all $\hat x \in \BBR^{\hat n}$, consider the control given by
\begin{subequations}\label{eq:qp_softmin.2}
\begin{align}
& \Big ( \hat u(\hat x), \hat \mu(\hat x) \Big ) \triangleq \underset{{\tilde u\in \BBR^m, \, \tilde \mu \in \BBR} }{\mbox{argmin}}  \, 
\hat J(\hat x, \tilde u) + \frac{1}{2}\gamma \tilde \mu^2 \label{eq:qp_softmin.a.2}\\
& \text{subject to}\nn\\
& L_{\hat f} \hat h(\hat x) + L_{\hat g} \hat h(\hat x) \tilde u + \alpha (\hat h(\hat x)) + \tilde \mu \hat h(\hat x) \ge 0. \label{eq:qp_softmin.b.2}
\end{align}
\end{subequations}
\Cref{lemma:qp_feas} applied to the cascade~\cref{eq:dynamics_aug.1,eq:dynamics_aug.2,eq:dynamics_aug.3} implies that if $L_{\hat g}\hat h$ is nonzero on $\hat \SB$, then the quadratic program \eqref{eq:qp_softmin.2} is feasible on $\BBR^{\hat n}$.
The following result provides a closed-form solution for the control $\hat u$ that satisfies \eqref{eq:qp_softmin.2}. 
This result is a corollary of \Cref{thm:lip}, which is obtained by applying \Cref{thm:lip} to the constrained optimization \eqref{eq:qp_softmin.2} and cascade~\cref{eq:dynamics_aug.1,eq:dynamics_aug.2,eq:dynamics_aug.3}.

\begin{corollary}{\rm
Assume that for all $\hat x \in \hat \SB$, $L_{\hat g} \hat h(\hat x) \ne 0$. 
Then,
\begin{equation}\label{eq:u_sol_ic}
    \hat u(\hat x) = \hat u_\rmd(\hat x) + L_{\hat g} \hat h(\hat x)^\rmT \hat \lambda(\hat x)
\end{equation}
where
\begin{equation}\label{eq:lambda_sol_ic}
    \hat \lambda(\hat x) =    
    \begin{cases}
     -\dfrac{L_{\hat f} \hat h(\hat x) + L_{\hat g} \hat h(\hat x) \hat u_\rmd(\hat x) + \alpha(\hat h(\hat x))}{L_{\hat g} \hat h(\hat x) L_{\hat g} \hat h(\hat x)^\rmT + \gamma^{-1} \hat h(\hat x)^2}, & \hat x\in \hat \Omega,\\
    0, & \hat x\notin \hat \Omega,
    \end{cases}
\end{equation}
and 
\begin{equation}\label{eq:Omega.2}
    \hat \Omega = \{\hat x\in\BBR^{\hat n}: L_{\hat f} \hat h(\hat x) + L_{\hat g} \hat h(\hat x)\hat u_\rmd(\hat x) +\alpha(\hat h(\hat x)) < 0\}.
\end{equation}
}\end{corollary}

\Cref{fig:block_diag} illustrates the architecture of the control~\Cref{eq:b_def.2,eq:h_def.2,eq:dynamics_control.a,eq:dynamics_control.b,eq:u_d_hat_def,eq:u_d_def,eq:Omega.2,eq:u_sol_ic,eq:lambda_sol_ic}. 
The following corollary is the main result on using this control for safety with input constraints.
This corollary is obtained by applying \Cref{thm:smocbf} to the cascade~\cref{eq:dynamics_aug.1,eq:dynamics_aug.2,eq:dynamics_aug.3}.

\begin{corollary}\label{cor:input_const}
\rm{
Consider~\eqref{eq:dynamics}, where~\ref{cond:hocbf.a} is satisfied, and consider $u$ given by~\Cref{eq:b_def.2,eq:h_def.2,eq:dynamics_control.a,eq:dynamics_control.b,eq:u_d_hat_def,eq:u_d_def,eq:Omega.2,eq:u_sol_ic,eq:lambda_sol_ic}, where~\ref{cond:hocbf_x_c.a} and \ref{cond:hocbf_x_c.d} are satisfied.
Assume that $\hat h^\prime$ is locally Lipschitz on $\hat \SSS$, and for all $\hat x \in \hat \SB$, $L_{\hat g}\hat h(\hat x) \ne 0$.
Let $\hat x_0\in \hat \SSS$.
Then, for all $t \in I(\hat x_0)$, $\hat x(t) \in \hat \SSS$, $x(t) \in \SSS_\rms$, and $u(t) \in \SU$.
}
\end{corollary}

\subsection{Ground Robot Example Revisited with Input Constraints}

We present examples to demonstrate the control~\Cref{eq:b_def.2,eq:h_def.2,eq:dynamics_control.a,eq:dynamics_control.b,eq:u_d_hat_def,eq:u_d_def,eq:Omega.2,eq:u_sol_ic,eq:lambda_sol_ic}.
The first example compares this control to an approach that uses multiple HOCBFs constraints in combination with input constraints.

\begin{example}\label{ex:compare}
\rm

We consider the nonholonomic ground robot from \Cref{example:nonholonomic_no_input_const}, and the map shown in \Cref{fig:feas_ex_traj}, which has one obstacle.
The area outside the obstacle is modeled by the zero-superlevel set of 
\begin{equation*}
    h_1(x) = 1 - \| [q_\rmx \quad q_\rmy]^\rmT \|_2 / 4.
\end{equation*}
The bounds on speed $v$ are modeled as the the zero-superlevel sets of 
\begin{equation*}
    h_2(x) = v+ 0.5, \qquad h_3(x) = v - 2.
\end{equation*}
We also consider input constraints. 
Specifically, the control must remain in the admissible set $\SU$ given by~\eqref{eq:U_def}, where 
\begin{align*}
    \phi_1(u) &= 2 - u_1, \qquad
    \phi_2(u) = u_1 + 2,\\
    \phi_3(u) &= 1 - u_2,\qquad
    \phi_4(u) = u_2 + 1,
\end{align*}
which implies $\SU = \{ u \in \BBR^2 \colon | u_1 | \le 2 \mbox{ and } |u_2| \le 1 \}$.
The desired control $u_\rmd$ and cost~\eqref{eq:quad_cost} are the same as in~\Cref{example:nonholonomic_no_input_const}.
Thus, $u_\rmd$ is given by \Cref{eq:u_d_ex,eq:u_d_1_ex,eq:u_d_2_ex,eq:psi}; and we consider the cost~\eqref{eq:quad_cost}, where $Q(x)= I_2$ and $c(x) = -u_\rmd(x)$, which is the minimum-intervention problem. 
We compare 2 control approaches.

The first approach is the soft-minimum R-CBF control~\Cref{eq:b_def.2,eq:h_def.2,eq:dynamics_control.a,eq:dynamics_control.b,eq:u_d_hat_def,eq:u_d_def,eq:Omega.2,eq:u_sol_ic,eq:lambda_sol_ic}. 
We use the control dynamics~\cref{eq:dynamics_control.a,eq:dynamics_control.b}, where $f_\rmc$, $g_\rmc$, and $h_\rmc$ are given by \Cref{ex:control dynamics} with $A_\rmc =\matls -1 &  0 \\ 0 & -0.6 \matrs$, $B_\rmc = -A_\rmc$, and $C_\rmc = I_2$.  
Thus, the control dynamics are LTI and low pass.
Conditions~\ref{cond:hocbf_x_c.a} and \ref{cond:hocbf_x_c.d} are satisfied with $d_\rmc = 1$ and $\zeta = 1$, and~\eqref{eq:hat d} implies that $\hat \ell = 7$, $\hat d_1 = 3$, $\hat d_2 = \hat d_3 = 2$, $\hat d_4 = \hat d_5= \hat d_6 = \hat d_7 = 1$. 
We implement the control~\Cref{eq:b_def.2,eq:h_def.2,eq:dynamics_control.a,eq:dynamics_control.b,eq:u_d_hat_def,eq:u_d_def,eq:Omega.2,eq:u_sol_ic,eq:lambda_sol_ic} with $\gamma = 100$, $\hat \alpha_{1,0}(h) =50h$, $\hat \alpha_{1,1}(h) = \hat \alpha_{2,0}(h) = \hat \alpha_{3,0}(h) = \alpha(h) = h$, and $x_{\rmc0}=[\,0\quad0\,]^\rmT$.
All other parameters are the same as in \Cref{example:nonholonomic_no_input_const}.

The second approach generates the control from a quadratic program with multiple HOCBF constraints and explicit input constraints (e.g., see~\cite{xiao2021high}). 
Specifically, the control $u$ and slack variables $\mu_1,\mu_2,\mu_3$ are computed by minimizing the cost 
\begin{equation*}
J_{\rm HOCBF}(u,\mu_1,\mu_2,\mu_3) = 
\| u - u_\rmd(x)\|_2^2 + \frac{1}{2}\sum_{j=1}^3 \gamma_j \mu_j^2,
\end{equation*}
subject $u \in \SU$ and 
\begin{align*}
L_f b_{1,1}(x) + L_g b_{1,1}(x) u + \alpha_1 (b_{1,1}(x)) + \mu_1 b_{1,1}(x) &\ge 0, \\
L_f h_2(x) + L_g h_2(x) u + \alpha_2 (h_2(x)) + \mu_2 h_2(x) &\ge 0, \\
L_f h_3(x) + L_g h_3(x) u + \alpha_3 (h_3(x)) + \mu_3 h_3(x) &\ge 0, 
\end{align*}
where $b_{1,1}$ is given by \eqref{eq:b_def} with $\alpha_{1,0}(h) =50h$, $\gamma_1 = \gamma_2 = \gamma_3 = 100$, and $\alpha_{1}(h) = \alpha_{2}(h) = \alpha_{3}(h) = h$.
Both methods are implemented at 50 Hz.

\Cref{fig:feas_ex_traj} shows the closed-loop trajectories for both methods with $x_0=[\,-3\quad -2\quad 1\quad \frac{\pi}{4}\,]^\rmT$ and the goal location $q_\rmd = [\,3\quad 3\,]^\rmT$, and \Cref{fig:feas_ex_u,fig:feas_ex_h} show the time histories of $h_j$ and $u$.
\Cref{fig:feas_ex_traj,fig:feas_ex_h} demonstrates that the multiple-HOCBF method does not satisfy safety; the robot collides with the obstacle at $t = 1.24$ s.
This failure occurs because the state and input constraints in the multiple-HOCBF method are completely independent.
Thus, the robot is able to enter a subset of the state space, where collision becomes unavoidable. 
As shown in \Cref{fig:feas_ex_u}, the controls saturate prior to the collision but there is insufficient control authority to avoid collision. 
In contrast, the soft-minimum R-CBF method uses the control dynamics and soft-minimum R-CBF to integrate all constraints and guarantee safety. 
As shown in \Cref{fig:feas_ex_traj,fig:feas_ex_u,fig:feas_ex_h}, the trajectory with the soft-minimum R-CBF method satisfies the safety and input constraints, and the robot reaches its destination.

\begin{figure}[ht!]
\center{\includegraphics[width=0.49\textwidth,clip=true,trim= 0.in 0.in 0.in 0.in] {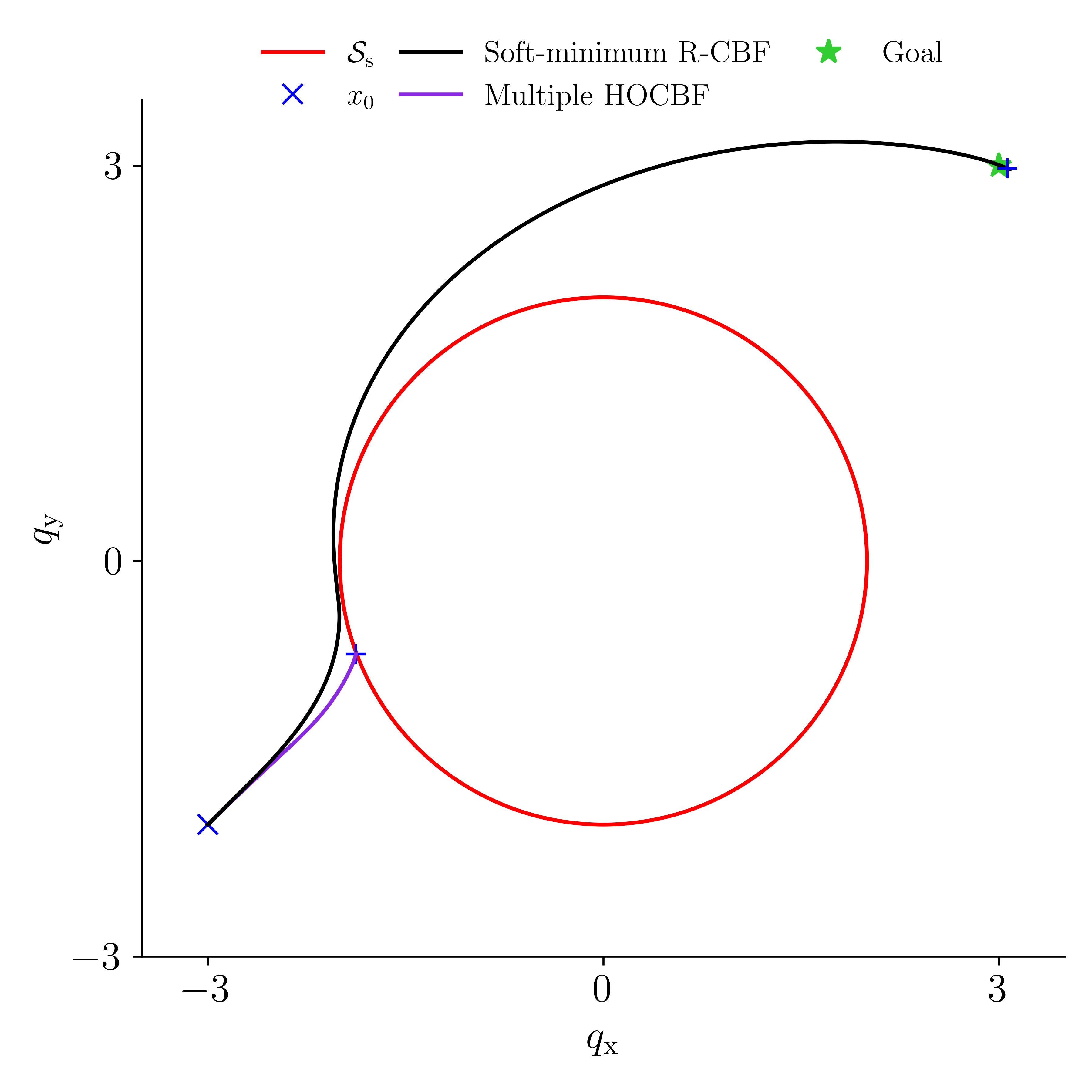}}
\caption{Closed-loop trajectories for soft-minimum R-CBF and multiple-HOCBF methods} \label{fig:feas_ex_traj}
\end{figure}

\begin{figure}[ht!]
\center{\includegraphics[width=0.49\textwidth,clip=true,trim= 0.in 0.in 0.in 0.in] {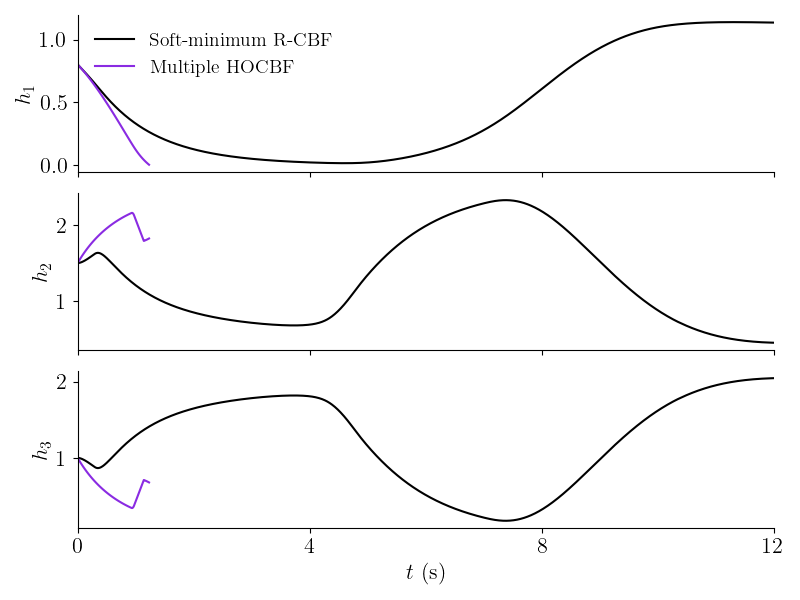}}
\caption{$h_1$, $h_2$, nd $h_3$ for soft-minimum R-CBF and multiple-HOCBF methods} \label{fig:feas_ex_h}
\end{figure}

\begin{figure}[ht!]
\center{\includegraphics[width=0.49\textwidth,clip=true,trim= 0.in 0.in 0.in 0.in] {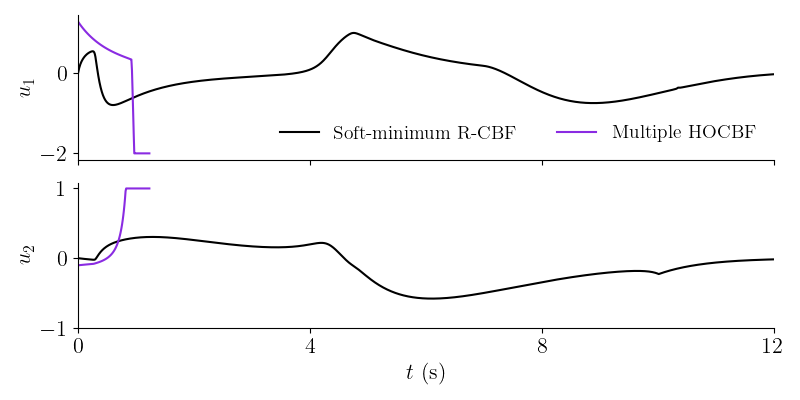}}
\caption{$u$ for soft-minimum R-CBF and multiple-HOCBF methods} \label{fig:feas_ex_u}
\end{figure}

In this simple example, it is possible to prevent collision with the multiple-HOCBF method by selecting more conservative values for $\gamma_1$, $\gamma_2$, $\gamma_3$, $\alpha_{1}$, $\alpha_{2}$, and $\alpha_{3}$.
However, this is problem specific, and it may not be possible to select conservative enough tuning in more complicated scenarios. 
Moreover, it is difficult to determine if a specific tuning is conservative enough to result in safe trajectories. 
In contrast, \Cref{cor:input_const} shows that the soft-minimum R-CBF method is safe for any choice of parameters. 
\exampletriangle
\end{example}

The next example revisits \Cref{example:nonholonomic_no_input_const} but includes input constraints as well as safety constraints.

\begin{example}\label{ex:input_const}\rm
We revisit the nonholonomic ground robot from~\Cref{example:nonholonomic_no_input_const} and include not only safety constraints but also input constraints. 
The safe set $\SSS_\rms$, desired control $u_\rmd$, and cost~\eqref{eq:quad_cost} are the same as in~\Cref{example:nonholonomic_no_input_const}.
Thus, $h_1,\ldots,h_9$ are given by \Cref{eq:h_obs,eq:h_wall,eq:h_speed.b}; $u_\rmd$ is given by \Cref{eq:u_d_ex,eq:u_d_1_ex,eq:u_d_2_ex,eq:psi}; and we consider the cost~\eqref{eq:quad_cost}, where $Q(x)= I_2$ and $c(x) = -u_\rmd(x)$, which is the minimum-intervention problem.

We also consider input constraints. 
Specifically, the control must remain in the admissibile set $\SU$ is given by~\eqref{eq:U_def}, where 
\begin{align*}
    \phi_1(u) &= 4 - u_1, \qquad
    \phi_2(u) = u_1 + 4,\\
    \phi_3(u) &= 1 - u_2,\qquad
    \phi_4(u) = u_2 + 1,
\end{align*}
which implies $\SU = \{ u \in \BBR^2 \colon | u_1 | \le 4 \mbox{ and } |u_2| \le 1 \}$ and $\hat \ell = 13$.

We use the low-pass LTI control dynamics from \Cref{ex:compare}, and~\eqref{eq:hat d} implies that $\hat d_1 = \cdots =\hat d_7=3$, $\hat d_8 = \hat d_9 = 2$, and $\hat d_{10} = \cdots = \hat d_{13} =1$.
We implement the control~\Cref{eq:b_def.2,eq:h_def.2,eq:dynamics_control.a,eq:dynamics_control.b,eq:u_d_hat_def,eq:u_d_def,eq:Omega.2,eq:u_sol_ic,eq:lambda_sol_ic} with $\rho = 10$, $\gamma = 200$, $\hat \alpha_{1,0}(h) = \ldots= \hat \alpha_{7,0}(h) = h$, $\hat \alpha_{1,1}(h) = \ldots= \hat \alpha_{6,1}(h) = 2.5h$, $\hat \alpha_{7,1}(h) = h$, $\hat \alpha_{8,0}(h) = \hat \alpha_{9,0}(h) = 10h$, $\alpha(h) = h$, and $\sigma_0 = 0.5$. 
For sample-data implementation, the control is updated at $50~\rm{Hz}$.

\Cref{fig:unicycle_hocbf_map_input_constraint} 2,500 closed-loop trajectories for $x_0=[\,-1\quad -8.5\quad 0\quad \frac{\pi}{2}\,]^\rmT$, $x_{\rmc 0}=[\,0\quad0\,]^\rmT$,
and randomly distributed goal locations. 
For each case, all safety and input constraints are satisfied along the closed-loop trajectory.
As discussed in \Cref{example:nonholonomic_no_input_const}, increasing the gains $k_1,k_2,k_3$ in the desired control \Cref{eq:u_d_ex,eq:u_d_1_ex,eq:u_d_2_ex,eq:psi} tends to enlarge the set of reachable goal locations.
Similar results where all state and input constraints are satisfied can be obtained with different choices of $\gamma >0$, $\alpha_{j,1},\ldots,\alpha_{j,d_j-2}$, and $\alpha$.
These user-selected values impact the aggressiveness of the control.

\begin{figure}[ht!]
\center{\includegraphics[width=0.49\textwidth,clip=true,trim= 0.in 0.in 0.0in 0.0in] {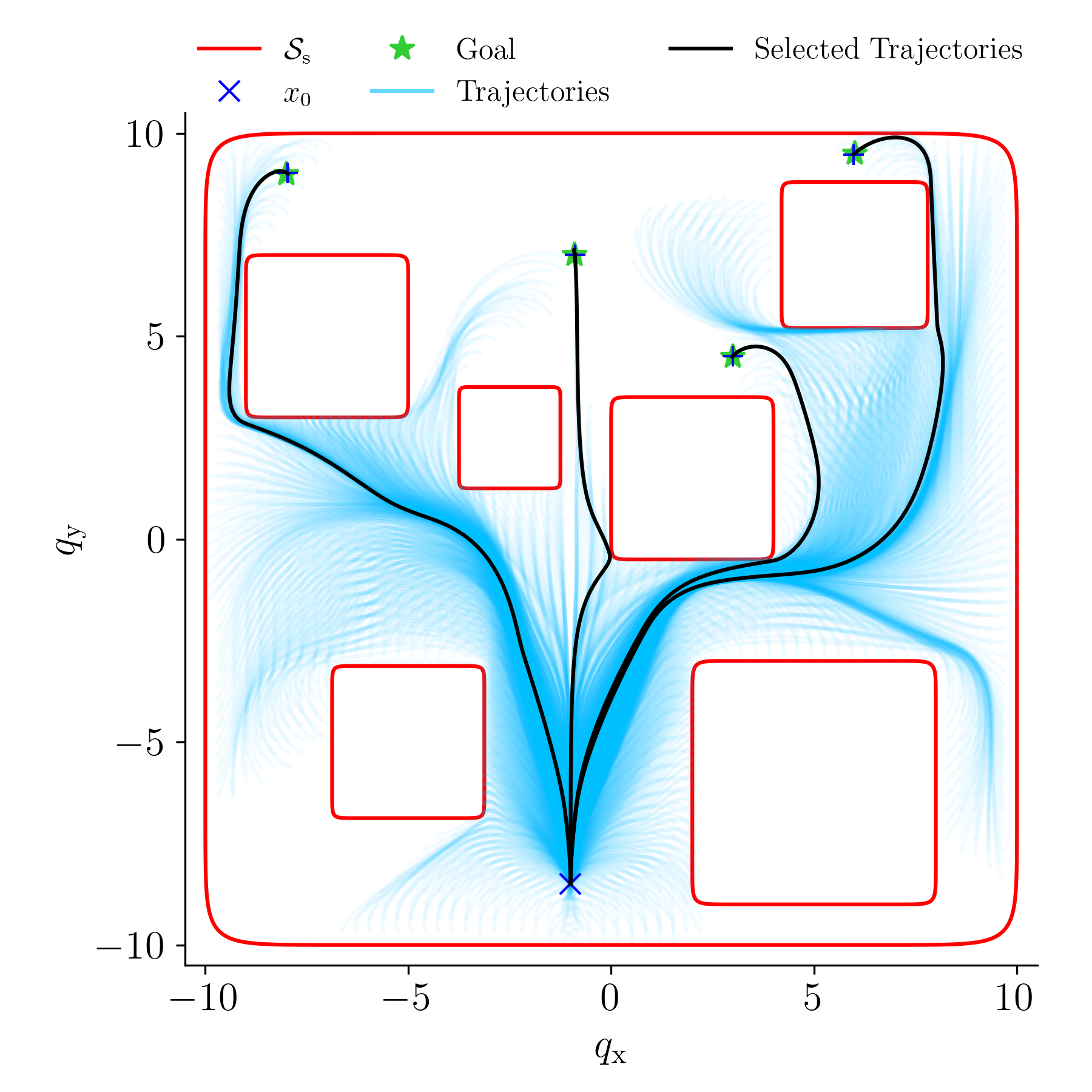}}
\caption{Safe set $\SSS_\rms$, and 2,500 closed-loop trajectories using the control~\Cref{eq:b_def.2,eq:h_def.2,eq:dynamics_control.a,eq:dynamics_control.b,eq:u_d_hat_def,eq:u_d_def,eq:Omega.2,eq:u_sol_ic,eq:lambda_sol_ic}.}\label{fig:unicycle_hocbf_map_input_constraint}
\end{figure}

\Cref{fig:unicycle_hocbf_map_input_constraint} also highlights the closed-loop trajectories for 
4 selected goal locations: $q_{\rmd}=[\,3\quad 4.5\,]^\rmT$, $q_{\rmd}=[\,-8\quad 9\,]^\rmT$, $q_{\rmd}=[\,6\quad 9.5\,]^\rmT$, and $q_{\rmd}=[\,-1\quad 7\,]^\rmT$.
Figures~\ref{fig:unicycle_hocbf_states_input_constraint} and~\ref{fig:unicycle_hocbf_h_input_constraint} show the trajectories of the relevant signals for the case where $q_{\rmd}=[\,3 \quad 4.5\,]^\rmT$.
Figure~\ref{fig:unicycle_hocbf_h_input_constraint}
shows that $\hat h$, $\min \hat b_{j,i}$, and $\min \hat h_{j}$ are positive for all time, which implies  that $x$ remains in $\SSS_\rms$ and $u$ remains in $\SU$.
\exampletriangle
\end{example}

\begin{figure}[ht!]
\center{\includegraphics[width=0.49\textwidth,clip=true,trim= 0.in 0.in 0in 0.in] {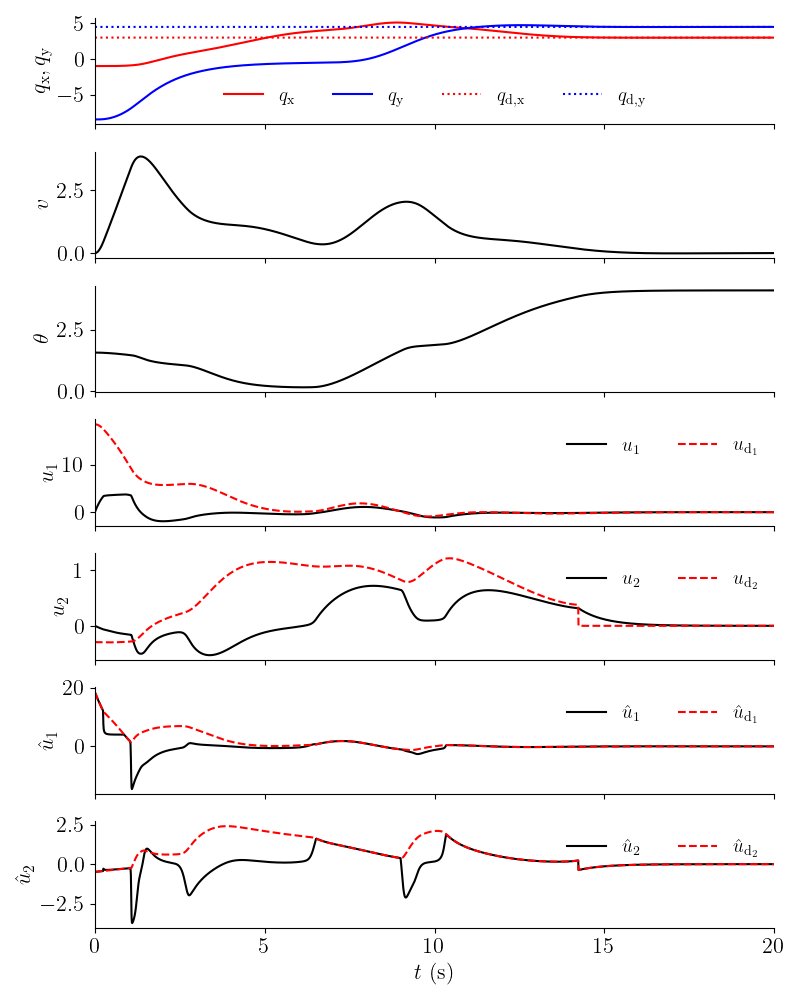}}
\caption{$q_\rmx$, $q_\rmy$, $v$, $\theta$, $u$, $u_\rmd$, $\hat u =    [ \, \hat u_1 \quad
    \hat u_2 \,]^\rmT$, and $\hat u_\rmd =    [ \, \hat u_{\rmd_1} \quad
    \hat u_{\rmd_2} \,]^\rmT$ for $q_{\rmd}=[\,3 \quad 4.5\,]^\rmT$.}\label{fig:unicycle_hocbf_states_input_constraint}
\end{figure}

\begin{figure}[ht!]
\center{\includegraphics[width=0.49\textwidth,clip=true,trim= 0.in 0.in 0.in 0.in] {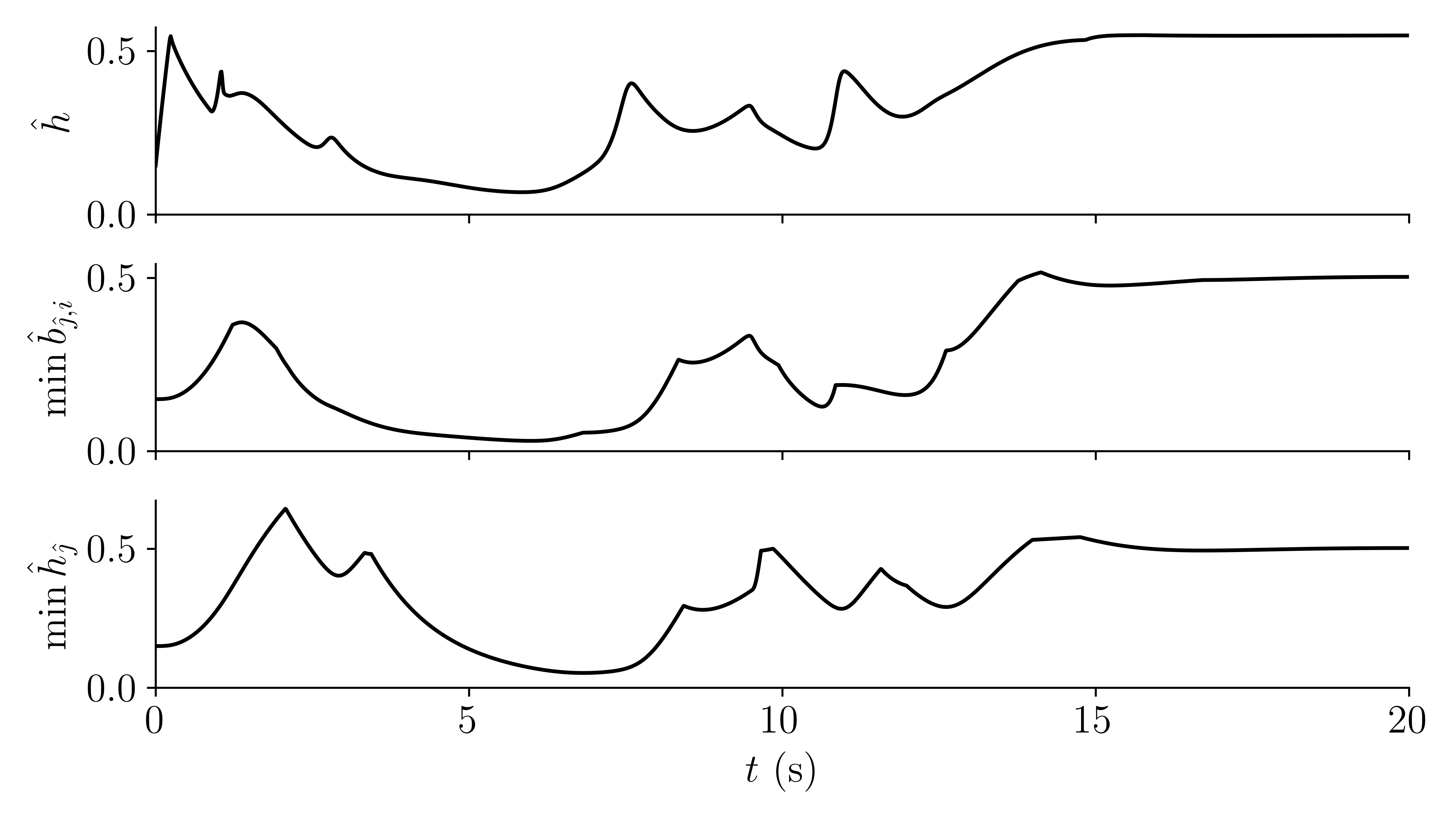}}
\caption{$\hat h$, $\min \hat b_{\hat \jmath,i}$, and $\min \hat h_{\hat \jmath}$ for $q_{\rmd}=[\,3 \quad 4.5\,]^\rmT$.} \label{fig:unicycle_hocbf_h_input_constraint}
\end{figure}

\section{Concluding Remarks}

This article presents several contributions. 
\Cref{section:composite CBF} presents and analyzes a method for constructing a composite soft-minimum R-CBF \eqref{eq:h_def} from multiple CBFs, which can have different relative degrees.
\Cref{prop:cont_fwd_S_C} is the main result of \Cref{section:composite CBF}, and it shows that the soft-minimum R-CBF describes a control forward invariant set $\SSS$, which is a subset of the safe set $\SSS_\rms$.
Next, \Cref{section:control} uses the composite soft-minimum R-CBF \eqref{eq:h_def} in a constrained quadratic optimization to construct a closed-form optimal control that guarantees safety. 
\Cref{thm:smocbf} is the main result of \Cref{section:control}, and this theorem shows that the closed-form optimal control guarantees safety.

\Cref{section:input constraints} presents the main result of this article, which is a closed-form optimal control that not only guarantees safety but also respects input constraints.
The key elements in the development of this control include the introduction of the control dynamics~\Cref{eq:dynamics_control.a,eq:dynamics_control.b} to transform the input constraint into state constraints; the use of the soft-minimum R-CBF to compose safety and input constraints, which have different relative degrees; and the introduction of a desired surrogate control \eqref{eq:u_d_hat_def} and associate surrogate cost \eqref{eq:surrogate cost}.
\Cref{cor:input_const} is the main result, which shows that the closed-form optimal control~\Cref{eq:b_def.2,eq:h_def.2,eq:dynamics_control.a,eq:dynamics_control.b,eq:u_d_hat_def,eq:u_d_def,eq:Omega.2,eq:u_sol_ic,eq:lambda_sol_ic} guarantees both safety- and input-constraint satisfaction.

One limitation of the method in this article is the assumption that $L_{\hat g} \hat h$ is nonzero on $\hat \SB$, which is a subset of the boundary of the zero-superlevel set of the soft-minimum R-CBF $\hat h$. 
In many instances, this condition can be verified numerical.
In addition, \Cref{prop:Lgh} provides a sufficient condition such that this assumption is satisfied.
However, determining easy to verify necessary and sufficient conditions is open for future work.
Moreover, an analysis of the restrictiveness of this assumption as compared to the typical CBF assumption \ref{cond:hocbf.a} is an open question.

The control for safety with input constraints depends on the choice of control dynamics \cref{eq:dynamics_control.a,eq:dynamics_control.b}. 
Safety is guaranteed for any control dynamics that satisfy \ref{cond:hocbf_x_c.a} and \ref{cond:hocbf_x_c.d}; however, the specific choice can impact performance. 
The examples in this article demonstrate that low-pass LTI control dynamics can be effective. 
It is reasonable to design the control dynamics as LTI and low pass or band pass based on the specific control objective (i.e., $u_\rmd$). 
However, other designs of the control dynamics \cref{eq:dynamics_control.a,eq:dynamics_control.b} may yield superior performance with respect to the cost $J$. 
Analyzing the design of \cref{eq:dynamics_control.a,eq:dynamics_control.b} based on the performance objective is an open question. 
Additionally, it may be possible to design \cref{eq:dynamics_control.a,eq:dynamics_control.b} in order to positively impact satisfaction of the assumption that $L_{\hat g} \hat h$ is nonzero on $\hat \SB$. 
This too is open for future work.

The approach to safety and input constraints can be directly extended to address input rate constraints (and constraints on higher-order time derivatives of the input). 
To accomplish this, the control dynamics~\Cref{eq:dynamics_control.a,eq:dynamics_control.b} are designed such that its relative degree $d_\rmc$ is greater than the positive integer $r$, where $\frac{\rmd^{r}}{\rmd t^{r}} u$ is the highest-order time derivative of the control that has a constraint. 
In this case, constraints on $u, \frac{\rmd}{\rmd t} u, \ldots, \frac{\rmd^{r}}{\rmd t^{r}} u$ are transformed into constraints on the controller state $x_\rmc$ using the approach in \Cref{section:input constraints}.

In this work, the soft minimum is used as a smooth approximation of the minimum. 
If $\rho>0$ is small, then the soft minimum is a conservative approximation of the minimum. 
In this case, $\hat \SH$ is a conservative approximation of $\cap_{j=1}^{\hat \ell} \hat \SC_{j,\hat d_j-1}$.
Thus, it is desirable to select $\rho>0$ large.
However, there are practical limits to the magnitude of $\rho$.
If $\rho >0$ is large, then $\| \hat h^\prime(x) \|_2$ is large at points where the minimum function is not differentiable, which can result in the surrogate control $\hat u$ having time derivatives with large magnitudes.  
Thus, selecting $\rho$ is a trade-off between the conservativeness of $\hat h$ and the size of $\| \hat h^\prime(x) \|_2$.

In this work, the log-sum-exponential soft minimum \eqref{eq:softmin} is used to compose multiple state constraints and multiple input constraints into a single constraint. 
In other words, the zero-superlevel set of the soft minimum \eqref{eq:softmin} is an approximation (subset) of the intersection of the zero-superlevel sets of the arguments of the soft minimum.
See \Cref{prop:S}.
Thus, the log-sum-exponential soft minimum can be used to approximate the intersection of zero-superlevel sets. 
Similarly, the log-sum-exponential soft maximum can be used to approximate the union of zero-superlevel sets. 
See \cite{rabiee2023softmax} for more details.

\appendices

\section{Proof of \Cref{prop:softmin_bound,prop:cascade_constraints,prop:cascade,proposition:error}} \label{app:input_constraint_propositions}

\begin{proof}[Proof of Proposition~\ref{prop:softmin_bound}]

Define $z_{\rm min} \triangleq \min \{z_1, z_2, \ldots, z_N\}$. 
Since the exponential is nonnegative and strictly increasing, it follows that 
\begin{equation*}
    e^{-\rho z_{\rm min}} \le \sum_{i=1}^N e^{-\rho z_i} \le N e^{-\rho z_{\rm min}},
\end{equation*}
and taking the logarithm yields
\begin{equation*}
    -\rho z_{\rm min} \le \log \sum_{i=1}^N e^{-\rho z_i} \le \log{N} - \rho z_{\rm min}.
\end{equation*}
Then, dividing by $-\rho$ yields the result.
\end{proof}

\begin{proof}[Proof of \Cref{prop:cascade_constraints}]

Let $t_1 \in I(\hat x_0)$. 
Since $\hat x(t_1) \in \hat \SSS_\rms$, it follows from \eqref{eq:hat Ss} that $x(t_1) \in \SSS_\rms$ and $x_\rmc(t_1) \in \SSS_\rmc$. 
Since $x_\rmc(t_1) \in \SSS_\rmc$, it follows from~\cref{eq:dynamics_control.b,eq:S_c_def} that $u(t_1) \in \SU$.
\end{proof}

\begin{proof}[Proof of \Cref{prop:cascade}]

Let $a\in\{1,2,\ldots, \ell\}$, and it follows from~\ref{cond:hocbf.a} and \ref{cond:hocbf_x_c.a} that~\ref{cond:app.a} in Appendix~\ref{app:appendix_cascade} is satisfied for $x_1$, $x_2$, $f_1$, $g_1$, $f_2$, $g_2$, $\SX_1$, $\SX_2$, $\xi_1$, $\xi_2$, $r_1$, $r_2$ equal to $x$, $x_\rmc$, $f$, $g$, $f_\rmc$, $g_\rmc$, $\BBR^n$, $\BBR^{n_\rmc}$, $h_a$, $h_\rmc$, $d_a$, $d_\rmc$, respectively. 
Thus, it follows from \Cref{thm:appendix} in Appendix~\ref{app:appendix_cascade} that $L_{\hat g} \hat h_a(\hat x) = \cdots =L_{\hat g} L_{\hat f}^{d_a+d_\rmc-2} \hat h_a(\hat x) = 0$ and $L_{\hat g} L_{\hat f}^{d_a+d_\rmc-1} \hat h_a(\hat x) = [ L_{g}L_{f}^{d_a-1}h_a(x)] L_{g_\rmc}L_{f_\rmc}^{d_\rmc - 1} h_\rmc(x_\rmc)$. 
Since, in addition,~\ref{cond:hocbf.a} implies that for all $x \in \SSS_\rms$, $L_{g}L_{f}^{d_a-1}h_a(x) \ne 0$ and \ref{cond:hocbf_x_c.a} implies that for all $x \in \BBR^{n_\rmc}$, $L_{g_\rmc}L_{f_\rmc}^{d_\rmc - 1} h_\rmc(x_\rmc)$ is nonsingular, it follows that for all $\hat x \in \hat \SSS_\rms$, $L_{\hat g} L_{\hat f}^{d_a+d_\rmc-1} \hat h_a(\hat x) \ne 0$, which confirms the result for $\hat \jmath \in\{1,2,\ldots, \ell\}$.

Let $b\in\{\ell+1,\ell+2,\ldots, \hat \ell\}$ and let $i \in \{0, 1, \ldots \zeta-1\}$.
Next, it follows from~\cref{eq:dynamics_aug.2} that
    $L_{\hat f}^{i} \hat h_b (\hat x) = L_{\hat f}^{i} \phi_{b- \ell} (h_\rmc(x_\rmc))
    = L_{\hat f}^{i-1} L_{f_\rmc} \phi_{b- \ell} (h_\rmc(x_\rmc)) = \cdots = L^{i}_{f_\rmc}  \phi_{b- \ell} (h_\rmc(x_\rmc))$, 
which combined with~\eqref{eq:dynamics_aug.3} implies that 
\begin{equation*}
L_{\hat g} L_{\hat f}^{i} \hat h_b(\hat x) = L_{\hat g} L^{i}_{f_\rmc} \phi_{b- \ell} (h_\rmc(x_\rmc))
=L_{g_\rmc} L^{i}_{f_\rmc} \phi_{b- \ell} (h_\rmc(x_\rmc)).
\end{equation*}
Thus, \ref{cond:hocbf_x_c.d} implies that for all $\hat x \in \BBR^{\hat n}$, $L_{\hat g} \hat h_b(\hat x) = \cdots = L_{\hat g} L_{\hat f}^{\zeta-2} \hat h_b(\hat x) = 0$ and for all $\hat x \in \hat \SSS_\rms$, $L_{\hat g} L_{\hat f}^{\zeta-1} \hat h_b(\hat x) \ne 0$, which combined with \eqref{eq:hat d} confirms the result for $\hat \jmath \in\{\ell+1,\ell+2,\ldots, \hat \ell\}$.
\end{proof}

\begin{proof}[Proof of \Cref{proposition:error}]

To prove~\ref{proposition:error.a}, it follows from \ref{cond:hocbf_x_c.a} that the $d_\rmc$th time derivative of $u$ along \cref{eq:dynamics_control.a,eq:dynamics_control.b} is
$u^{(d_\rmc)} = L_{f_\rmc}^{d_\rmc} h_\rmc(x_\rmc) + L_{g_\rmc} L_{f_\rmc}^{d_\rmc - 1}h_\rmc(x_\rmc)\hat u(\hat x)$.
Since, in addition, $\hat u = \hat u_\rmd$ and $\sigma_{d_\rmc}=1$, substituting \eqref{eq:u_d_hat_def} yields
\begin{equation}\label{eq:ud_deriv}
\frac{\rmd^{d_\rmc}}{\rmd t^{d_\rmc}} u = L_{\hat f}^{d_\rmc} u_\rmd(\hat x) +  \sum_{i = 0}^{d_\rmc-1} \sigma_i \Big(L_{\hat f}^{i} u_\rmd(\hat x)-L_{f_\rmc}^i h_\rmc(x_\rmc)\Big).
\end{equation}
Next,~\ref{cond:hocbf.a} and \ref{cond:hocbf_x_c.a} imply that~\ref{cond:app.a} in Appendix~\ref{app:appendix_cascade} is satisfied for $x_1$, $x_2$, $f_1$, $g_1$, $f_2$, $g_2$, $\SX_1$, $\SX_2$, $\xi_1$, $\xi_2$, $r_1$, $r_2$ equal to $x$, $x_\rmc$, $f$, $g$, $f_\rmc$, $g_\rmc$, $\BBR^n$, $\BBR^{n_\rmc}$, $h_{\hat \jmath}$, $h_\rmc$, $d_{\hat \jmath}$, $d_\rmc$, respectively.
Thus, it follows from~\Cref{lem:app_func_x1} in Appendix~\ref{app:appendix_cascade} with $\hat \nu = u_\rmd$ that $L_{\hat g} u_\rmd(\hat x) = L_{\hat g} L_{\hat f} u_\rmd(\hat x) = \cdots = L_{\hat g} L_{\hat f}^{d_\rmc -1 } u_\rmd(\hat x) = 0$.
Thus, for $i \in \{0, 1, \ldots, d_\rmc \}$, the $i$th time derivative of~\eqref{eq:u_d_def} along~\cref{eq:dynamics_aug.1,eq:dynamics_aug.2,eq:dynamics_aug.3} is
\begin{equation}\label{eq:ud_hat_deriv}
\frac{\rmd^i}{\rmd t^i} u_\rmd(\hat x) = L_{\hat f}^i u_\rmd(\hat x).  
\end{equation}
Similarly, for $i \in \{0, 1, \ldots, d_\rmc -1\}$, it follows from \ref{cond:hocbf_x_c.a} that the $i$th time derivative of~\eqref{eq:dynamics_control.b} along~\cref{eq:dynamics_control.a} is
\begin{equation}\label{eq:u_hat_deriv}
  \frac{\rmd^i}{\rmd t^i} u = L_{f_\rmc}^i h_\rmc(x_\rmc).
\end{equation}
Finally, substituting~\cref{eq:ud_hat_deriv,eq:u_hat_deriv} into~\eqref{eq:ud_deriv} yields~\eqref{eq:sum_ue}, which confirms~\ref{proposition:error.a}.

To prove~\ref{proposition:error.b}, since all roots of $\sigma$ are in the open left-hand complex plane, it follows from \eqref{eq:sum_ue} that for all $\hat x_0 \in \BBR^{\hat n}$, $\lim_{t \to \infty} \left [ u(t) - u_\rmd(\hat x(t)) \right ] =0$ exponentially, which confirms~\ref{proposition:error.b}.

To prove~\ref{proposition:error.c}, define $u_\rme \triangleq u - u_\rmd(\hat x)$.
Thus, $u=u_\rme+u_\rmd(\hat x)$, and it follows from \Cref{eq:quad_cost,eq:u_d_def} that
\begin{align*}
J(x, u) 
&= \frac{1}{2}(u_\rme +u_\rmd(\hat x))^\rmT Q(x) (u_\rme +u_\rmd(\hat x)) \nn \\
   &\qquad + c(x)^\rmT (u_\rme +u_\rmd(\hat x))\\
    &= \frac{1}{2}u_\rme^\rmT Q(x) u_\rme + u_\rmd(\hat x)^\rmT Q(x) u_\rme + c(x)^\rmT u_\rme \\
    &\qquad + \frac{1}{2} u_\rmd(\hat x)^\rmT Q(x) u_\rmd(\hat x) + c(x)^\rmT u_\rmd(\hat x)\\
    &= \frac{1}{2}u_\rme^\rmT Q(x) u_\rme + J(x,u_\rmd(\hat x))
\end{align*}
which implies that $J(x,u) - J(x,u_\rmd(\hat x)) =\frac{1}{2} u_\rme^\rmT Q(x) u_\rme$.
Since, in addition, $\lim_{t \to \infty} u_\rme(t) = 0$ and $Q$ is bounded, it follows that $\lim_{t \to \infty} [ J(x(t),u(t)) - J(x(t),u_\rmd(\hat x(t))) ] =0$, which confirms~\ref{proposition:error.c}.

To prove~\ref{proposition:error.d}, since for $i\in\{0, 1, \ldots, d_\rmc-1\}$, $L_{f_\rmc}^i h_\rmc(x_{\rmc 0}) = L_{\hat f}^i u_\rmd(\hat x_0)$, it follows from~\cref{eq:ud_hat_deriv,eq:u_hat_deriv} that for $i\in\{0, 1, \ldots, d_\rmc-1\}$, $  \frac{\rmd^i}{\rmd t^i} u(t) |_{t=0} = \frac{\rmd^i}{\rmd t^i}  u_\rmd(\hat x(t)) |_{t=0}$, which combined with \eqref{eq:sum_ue}, implies that $u(t) \equiv u_\rmd(\hat x(t))$.
\end{proof}

\section{Relative Degree of a Nonlinear Cascade} \label{app:appendix_cascade}

This appendix examines the relative degree of a cascade of nonlinear systems.
The results in this appendix are needed for \Cref{prop:cascade}
and \Cref{proposition:error} in \Cref{section:input constraints}. 
Consider
\begin{align}
    \dot x_1(t) &= f_1(x_1(t)) + g_1(x_1(t))u_1(t), \label{eq:appendix_dyn1}\\
    \dot x_2(t) &= f_2(x_2(t)) + g_2(x_2(t))u_2(t), \label{eq:appendix_dyn2}
\end{align}
where for $i\in \{ 1,2 \}$, $x_i(t) \in \BBR^{n_i}$ is the state, $x_i(0) = x_{i0} \in \BBR^{n_i}$ is the initial condition, and $u_i(t) \in \BBR^{m_i}$ is the input.

Let $\xi_1:\BBR^{n_1} \to \BBR^{l_1}$ and $\xi_2: \BBR^{n_2} \to \BBR^{m_1}$.
We make the following assumption: 
\begin{enumA}[resume]
    \item \label{cond:app.a}
    For $i\in \{ 1,2 \}$, there exists $\SX_i \subseteq \BBR^{n_i}$ and a positive integer $r_i$ such that for all $x_i \in \SX_i$, $L_{g_i} \xi_i(x_i) = L_{g_i} L_{f_i} \xi_i(x_i) = \cdots = L_{g_i} L_{f_i}^{r_i-2} \xi_i(x_i)= 0$.
\end{enumA}
Assumption~\ref{cond:app.a} implies that if $a_i \in \SX_i$ and $L_{g_i} L_{f_i}^{r_i-1}\xi_i(a_i) \ne 0$, then $\xi_i$ has relative degree $r_i$ with respect to $\dot x_i = f_i(x_i) + g_i(x_i) u_i$ at $a_i$.

Next, we consider the cascade of \eqref{eq:appendix_dyn1} and \eqref{eq:appendix_dyn2}, where $u_1 = \xi_2(x_2)$, which is given by 
\begin{align}
    \dot{\hat x} &= \hat f(\hat x) + \hat g(\hat x) u_2,\label{eq:cascase_appendix.a}\\
    y &= \hat h(\hat x),\label{eq:cascase_appendix.b}
\end{align}
where 
\begin{gather}
\hat x \triangleq \begin{bmatrix} x_1\\ x_2\end{bmatrix},
\qquad 
\hat f(\hat x) \triangleq \begin{bmatrix}
            f_1(x_1) + g_1(x_1)\xi_2(x_2)\\
            f_2(x_2)
        \end{bmatrix}, \label{eq:cascade_defs.a}\\ 
\hat g(\hat x) \triangleq \begin{bmatrix}
            0\\
            g_2(x_2)
        \end{bmatrix}, \label{eq:cascade_defs.b}
\qquad 
\hat h(\hat x) \triangleq \xi_1(x_1).
\end{gather}

The following preliminary results are needed.

\begin{lemma}\label{lem:Lgh_upto_d1}
\rm{
Consider~\cref{eq:cascase_appendix.a,eq:cascase_appendix.b,eq:cascade_defs.a,eq:cascade_defs.b,eq:appendix_dyn1,eq:appendix_dyn2}, where~\ref{cond:app.a} is satisfied.
For all $\hat x \in \SX_1 \times \SX_2$, the following statements hold:
\begin{enumalph}
\item \label{app.lem.1.1}
Let $j \in \{0,1,\ldots, r_1-1\}$. 
Then, $L_{\hat f}^j \hat h(\hat x) = L_{f_1}^j \xi_1(x_1)$. 

\item \label{app.lem.1.2}
Let $j \in \{0,1,\ldots, r_1-2\}$. 
Then, $L_{\hat g} L_{\hat f}^j \hat h(\hat x) = 0$.
\end{enumalph}
}
\end{lemma}

\begin{proof}

To prove \ref{app.lem.1.1}, we use induction on $j$. 
First, note that $L_{\hat f}^0 \hat h(\hat x) = \hat h(\hat x) = \xi_1(x_1) = L_{f_1}^0 \xi_1(x_1)$, which implies that \ref{app.lem.1.1} holds for $j=0$. 
Next, assume that \ref{app.lem.1.1} holds for $j=a\in \{0,1,\ldots,r_1-2\}$. 
Thus, 
\begin{equation*}
L_{\hat f}^{a+1} \hat h(\hat x) = L_{\hat f} L_{\hat f}^{a} \hat h(\hat x)= L_{\hat f} L_{f_1}^{a} \xi_1(x_1),    
\end{equation*}
which combined with \eqref{eq:cascade_defs.a} and \ref{cond:app.a} yields
\begin{align*}
L_{\hat f}^{a+1} \hat h(\hat x) 
&= L_{f_1}^{a+1} \xi_1(x_1) + [ L_{g_1}L_{f_1}^{a} \xi_1(x_1) ]\xi_2(x_2)\\ 
&=L_{f_1}^{a+1} \xi_1(x_1),
\end{align*}
which confirms \ref{app.lem.1.1}.

To prove \ref{app.lem.1.2}, let $b \in \{0,1,\ldots,r_1-2\}$, and it follows from \ref{app.lem.1.1} that $L_{\hat g} L_{\hat f}^{b} \hat h(\hat x) = L_{\hat g} L_{f_1}^{b} \xi_1(x_1) = 0$.
\end{proof}

Let $\nu: \BBR^{n_1} \to \BBR^{\ell_1}$ be continuously differentiable, and let $\hat \nu:\BBR^{n_1 + n_2} \to \BBR^{\ell_1}$ be defined by $\hat \nu(\hat x) \triangleq \nu(x_1)$.

\begin{lemma}\label{lem:Lf_struct}
\rm{
Let $j$ be a positive integer. 
Then, there exists $F_j:\BBR^{n_1\times m^{j-1}} \to \BBR^{\ell_1}$ such that for all $\hat x \in \SX_1 \times \SX_2$,
\begin{align}\label{lem:rec_Lf}
L_{\hat f}^j\hat \nu(\hat x) &= F_j \left (x_1, \xi_2(x_2), L_{f_2} \xi_2(x_2),\ldots, L_{f_2}^{j-2} \xi_2(x_2) \right )\nn\\
&\qquad +  \left [ L_{g_1}\nu(x_1) \right ] L_{f_2}^{j-1} \xi_2(x_2).  
\end{align}
}
\end{lemma}

\begin{proof}

We use induction on $j$. 
First,~\eqref{eq:cascade_defs.a} implies that $L_{\hat f} \hat \nu (\hat x) = L_{\hat f} \nu(x_1) = L_{f_1} \nu(x_1) + [ L_{g_1} \nu(x_1) ] \xi_2(x_2) = F_1(x_1) + [ L_{g_1} \nu(x_1) ] \xi_2(x_2)$, where $F_1(x_1) = L_{f_1} \nu(x_1)$, which confirms \eqref{lem:rec_Lf} for $j = 1$.

Next, assume that \eqref{lem:rec_Lf} holds for $j = a  \in \{1, 2, \ldots\}$.
Thus, 
\begin{align}
    L_{\hat f}^{a+1}\hat \nu(\hat x) &= L_{\hat f} L_{\hat f}^{a}\hat \nu(\hat x) \nn\\
    &= L_{\hat f} F_{a} + L_{\hat f} \left[L_{g_1}\nu(x_1) L_{f_2}^{a-1} \xi_2(x_2)\right], \label{lem:Lf_struct_1}
\end{align}
where the arguments of $F_{a}$ are omitted.
Next, note that it follows from \eqref{eq:cascade_defs.a} that 
\begin{align}
    L_{\hat f} F_{a} &= \pderiv{F_{a}}{x_1} \left [f_1(x_1) + g_1(x_1)\xi_2(x_2)\right]\nn\\
    &\qquad + \sum_{k=0}^{a-2} \pderiv{F_a}{L_{f_2}^k \xi_2} \pderiv{L_{f_2}^k \xi_2(x_2)}{x_2} f_2(x_2)\nn\\
    &= L_{f_1} F_{a} + \left [ L_{g_1} F_{a} \right ]\xi_2(x_2)+ \sum_{k=0}^{a-2} \pderiv{F_j}{L_{f_2}^k \xi_2} L_{f_2}^{k+1}\xi(x_2).\label{lem:Lf_struct_2}
\end{align}
and
\begin{align}\label{lem:Lf_struct_3}
L_{\hat f}\left[L_{g_1}\nu(x_1) L_{f_2}^{a-1} \xi_2(x_2)\right] 
&= \left ( I_{\ell_1} \otimes  L_{f_2}^{a -1}\xi_2(x_2) \right )^\rmT G(x_1) \nn \\
& \qquad \times \left ( f_1(x_1) + g_1(x_1) \xi_2(x_2) \right )\nn\\
&\qquad+ L_{g_1} \nu(x_1) L_{f_2}^{a} \xi_2(x_2),
\end{align}
where $\otimes$ is the Kronecker product, 
\begin{equation*}
    G(x_1) \triangleq \begin{bmatrix}
        \pderiv{}{x_1}  \left[L_{g_1}\nu(x_1)\right]_{(1)} ^\rmT\\
        \vdots\\
        \pderiv{}{x_1} \left[L_{g_1}\nu(x_1)\right]_{(\ell_1)}^\rmT\\
    \end{bmatrix},
\end{equation*}
and $\left[L_{g_1}\nu(x_1)\right]_{(i)}$ is the $i$th row of $L_{g_1}\nu(x_1)$.
Thus, substituting~\cref{lem:Lf_struct_2,lem:Lf_struct_3} into \eqref{lem:Lf_struct_1} yields
\begin{equation*}
L_{\hat f}^{a+1} \hat \nu(\hat x) = F_{a+1}+  L_{g_1}\nu(x_1) L_{f_2}^{a} \xi_2(x_2),  
\end{equation*}
where
\begin{align*}
F_{a+1} &= L_{f_1} F_{a} + \left [ L_{g_1} F_{a} \right ]\xi_2(x_2)+ \sum_{k=0}^{a-2} \pderiv{F_j}{L_{f_2}^k \xi_2} L_{f_2}^{k+1}\xi(x_2)\\
&\qquad +\left ( I_{\ell_1} \otimes  L_{f_2}^{a -1}\xi_2(x_2) \right )^\rmT G(x_1)\\
& \qquad \times \left ( f_1(x_1) + g_1(x_1) \xi_2(x_2) \right ),
\end{align*}
which confirms \eqref{lem:rec_Lf} for $j=a+1$.
\end{proof}

\begin{lemma} \label{lem:app_func_x1}
\rm{
Consider~\cref{eq:appendix_dyn1,eq:appendix_dyn2,eq:cascase_appendix.a,eq:cascase_appendix.b,eq:cascade_defs.a,eq:cascade_defs.b}, where~\ref{cond:app.a} is satisfied. For all $\hat x\in \SX_1 \times \SX_2$, the following statements hold:
\begin{enumalph}

\item \label{app_func_x1.1}
For $j\in\{0,1,\ldots,r_2 -1\}$, $L_{\hat g} L_{\hat f}^{j} \hat \nu(\hat x) = 0$. 

\item \label{app_func_x1.2}
$L_{\hat g} L_{\hat f}^{r_2} \hat \nu (\hat x)= \left [ L_{g_1} \nu(x_1) \right ] L_{g_2} L_{f_2}^{r_2 -1} \xi_2(x_2)$.
\end{enumalph}
}
\end{lemma}

\begin{proof}

It follows from Lemma \ref{lem:Lf_struct} that for all positive integers $a$, 
\begin{align}\label{eq:app_func_x1.int1}
L_{\hat g} L_{\hat f}^a \hat{\nu}(\hat x) &= L_{\hat g} F_a(x_1, L_{f_2}^0 \xi_2(x_2), \ldots, L_{f_2}^{a-2} \xi_2(x_2)) \nn\\
&\qquad + L_{\hat g}\left( \left [ L_{g_1}\nu(x_1) \right ] L_{f_2}^{a-1} \xi_2(x_2)\right).
\end{align}

To prove \ref{app_func_x1.1}, let $j \in \{0,1,\ldots,r_2-1\}$. 
Note that~\eqref{eq:cascade_defs.b} and~\ref{cond:app.a} imply that $L_{\hat g} F_j(x_1, \xi_2(x_2), \ldots, L_{f_2}^{j-2} \xi_2(x_2)) =0$ and $L_{\hat g}(  [ L_{g_1}\nu(x_1) ] L_{f_2}^{j-1} \xi_2(x_2)) =0$, which together with \eqref{eq:app_func_x1.int1} confirms \ref{app_func_x1.1}.

To prove \ref{app_func_x1.2}, it follows from~\eqref{eq:cascade_defs.b} and~\ref{cond:app.a} that $L_{\hat g} F_{r_2}(x_1, \xi_2(x_2), \ldots, L_{f_2}^{r_2-2} \xi_2(x_2)) =0$ and $L_{\hat g}( [ L_{g_1}\nu(x_1) ] L_{f_2}^{r_2-1} \xi_2(x_2)) =[ L_{g_1}\nu(x_1) ] L_{g_2} L_{f_2}^{r_2-1} \xi_2(x_2)$, which together with \eqref{eq:app_func_x1.int1} confirms \ref{app_func_x1.1}.
\end{proof}

The next theorem is the main result on the relative degree of a cascade of nonlinear systems.
The result shows that if $\xi_i$ has relative degree $r_i$ with respect to $\dot x_i = f_i(x_i) + g_i(x_i) u_i$ at $a_i \in \SX_i$, then the relative degree of the cascade~\cref{eq:cascase_appendix.a,eq:cascase_appendix.b,eq:cascade_defs.a,eq:cascade_defs.b} is greater than or equal to $r_1+r_2$.
Furthermore, the relative degree of the cascade~\cref{eq:cascase_appendix.a,eq:cascase_appendix.b,eq:cascade_defs.a,eq:cascade_defs.b} is $r_1+r_2$ if and only if $[L_{g_1}L_{f_1}^{r_1-1}h_1(a_1)]L_{g_2}L_{f_2}^{r_2 - 1}h_2(a_2)$ is nonzero.

\begin{theorem}\label{thm:appendix}
\rm{
Consider~\cref{eq:cascase_appendix.a,eq:cascase_appendix.b,eq:cascade_defs.a,eq:cascade_defs.b,eq:appendix_dyn1,eq:appendix_dyn2}, where~\ref{cond:app.a} is satisfied. 
Then, for all $\hat x \in \SX_1 \times \SX_2$, the following statements hold:
\begin{enumalph}
    \item \label{thm:appendix.a}
    For all $j \in \{0, 1, \ldots, r_1 +r_2-2\}$, $L_{\hat g} L_{\hat f}^{j} \hat h(\hat x) = 0$.

    \item \label{thm:appendix.b}
    $L_{\hat g} L_{\hat f}^{r_1+r_2-1} \hat h(\hat x) = [ L_{g_1}L_{f_1}^{r_1-1}\xi_1(x_1)] L_{g_2}L_{f_2}^{r_2 - 1}\xi_2(x_2)$. 
\end{enumalph}
}
\end{theorem}

\begin{proof}

Define $\nu(x_1) \triangleq L_{f_1}^{r_1-1}\xi_1(x_1)$ and $\hat \nu(\hat x) \triangleq \nu(x_1)$.

To prove \ref{thm:appendix.a}, it follows from Lemma \ref{lem:Lgh_upto_d1} that for all $j \in \{0,\ldots, r_1-2\}$, $L_{\hat g} L_{\hat f}^{j} \hat h(\hat x) = 0$. 

Next, let $a \in \{r_1-1,r_1,\ldots,r_1 + r_2-2\}$ and define $b \triangleq a - r_1 +1$. 
Thus, $L_{\hat g}L_{\hat f}^a \hat h(\hat x) = L_{\hat g}L_{\hat f}^b L_{\hat f}^{r_1-1} \hat h(\hat x) = L_{\hat g} L_{\hat f} ^b \hat \nu(\hat x)$.
Since, in addition, $b \in \{0,\ldots,r_2-1 \}$, it follows from Lemma \ref{lem:app_func_x1} that $L_{\hat g} L_{\hat f}^a \hat h(\hat x) = 0$, which implies that for all $j \in \{r_1-1,r_1,\ldots,r_1 + r_2-2\}$, $L_{\hat g} L_{\hat f}^{j} \hat h(\hat x) = 0$.

To prove \ref{thm:appendix.b}, note that $L_{\hat g} L_{\hat f}^{r_1+r_2-1} \hat h(\hat x) = L_{\hat g} L_{\hat f}^{r_2} L_{\hat f}^{r_1-1} \hat h(\hat x) = L_{\hat g} L_{\hat f}^{r_2} \hat \nu(\hat x)$.
Thus,~\Cref{lem:app_func_x1} implies that $L_{\hat g} L_{\hat f}^{r_1+r_2-1} \hat h(\hat x) = [L_{g_1} \nu(x_1)] L_{g_2} L_{f_2}^{r_2 -1} \xi_2(x_2) = [L_{g_1} L_{f_1}^{r_1-1}\xi_1(x_1)] L_{g_2} L_{f_2}^{r_2 -1} \xi_2(x_2)$. 
\end{proof}

 \bibliographystyle{IEEEtran} 
 \bibliography{softmin_hocbf}
 
\end{document}